\documentclass[journal]{IEEEtran}
\usepackage{algorithm}

\usepackage{algpseudocode}
\usepackage{moreverb}
\usepackage{amsmath}
\usepackage{colortbl}
\usepackage{wrapfig}
\usepackage{bm}
\usepackage{tabularx}
\usepackage{cite}
\usepackage{graphicx}
\usepackage{amsfonts}
\usepackage{amsthm}
\usepackage{amssymb}
\usepackage{cases}
\usepackage{bbm}
\usepackage{stfloats}
\usepackage{subcaption}
\usepackage{xcolor}
\usepackage{xpatch}
\usepackage{epstopdf}
\usepackage[justification=centering]{caption}
\newtheorem{remark}{Remark}
\usepackage{balance}

\makeatletter
\ExplSyntaxOn

\ExplSyntaxOff

\newcommand\bib@setcolor[1]{%
	\ifcsname bib@colored@#1\endcsname
	\expandafter\color\expandafter{\csname bib@colored@#1\endcsname}
	\else
	\normalcolor
	\fi
}

\xpatchcmd\@bibitem
{\item}
{\bib@setcolor{#1}\item}
{}{\fail}

\xpatchcmd\@lbibitem
{\item}
{\bib@setcolor{#2}\item}
{}{\fail}
\makeatother

\ifCLASSOPTIONcaptionsoff
\usepackage[nomarkers]{endfloat}
\let\MYoriglatexcaption\caption
\renewcommand{\caption}[2][\relax]{\MYoriglatexcaption[#2]{#2}}
\fi
\usepackage{color}

\newtheorem{Theorem}{Theorem}

\allowdisplaybreaks[4]
\usepackage[
colorlinks=false,     
pdfborder={0 0 0},    
hidelinks             
]{hyperref}

\begin{document}

\title{\huge Autonomous Driving with RSMA-Enabled Finite Blocklength Transmissions: Ergodic Performance Analysis and Optimization}

\author{Yi~Wang,
~Yingyang~Chen,~\IEEEmembership{Senior Member,~IEEE}
,~Li~Wang,~\IEEEmembership{Senior Member,~IEEE}, \\Donghong Cai, \IEEEmembership{Member,~IEEE}, Xiaofan Li, \IEEEmembership{Senior Member,~IEEE}, and Pingzhi~Fan,~\IEEEmembership{Fellow,~IEEE}

\thanks{Y. Wang and Y. Chen are with the Department of Electronic Engineering, College of Information Science and Technology, Jinan University, Guangzhou 510632, China (e-mail: yiwang@stu2023.jnu.edu.cn; chenyy@jnu.edu.cn).}
\thanks{L. Wang is with the School of Computer Science, Beijing University of Posts
and Telecommunications, Beijing 100876, China (e-mail: liwang@bupt.edu.cn).}
\thanks{D. Cai is with the College of Cyber Security and College of Information Science and Technology, Jinan University, Guangzhou 510632, China (e-mail: dhcai@jnu.edu.cn).}
\thanks{X. Li is with the School of Intelligent Systems Science and Engineering, Jinan University, Zhuhai 519070, China (e-mail: lixiaofan@jnu.edu.cn).}
\thanks{P. Fan is with the Key Laboratory of Information Coding and Transmission, Southwest Jiaotong University, Chengdu 610031, China (e-mail: p.fan@ieee.org).}
    
 }

\markboth{}%
{}


\maketitle

\begin{abstract}
Rate-splitting multiple access (RSMA) is a key technology for next-generation multiple access systems due to its robustness against imperfect channel state information (CSI). This makes RSMA particularly suitable for high-mobility autonomous driving, where ultra-reliable and low-latency communication (URLLC) is essential. To address the stringent requirements, this study enables RSMA finite blocklength (FBL) transmissions and explicitly evaluates the ergodic performance. We derive the closed-form lower bound for the ergodic sum-rate of RSMA, considering vital factors such as the vehicle velocities, vehicle positions, power allocation of each stream, blocklengths, and block error rates (BLERs). To further enhance the ergodic sum-rate while complying with quality of service (QoS) rate constraints, we jointly optimize the global power coefficient, private power distribution, and common rate splitting. Guided by gradient descent, we first adjust the global power coefficient based on its sum-rate solution. This parameter regulates the power state of the common stream, allowing for dynamic activation or deactivation: if active, we optimize the private power distribution and adjust the common rate splitting to meet minimum transmission constraints; if inactive, we use the sequential quadratic programming for private power distribution optimization. Simulation results confirm that our RSMA scheme significantly improves the ergodic performance, reduces blocklength and BLER, surpassing the RSMA counterpart with average private power and space division multiple access (SDMA). Furthermore, our approach is validated to guarantee the rates for users with the poorest channel conditions, thereby enhancing fairness across the network.
\end{abstract}
\vspace{-1mm}
\vspace{-2mm}

%

\section{Introduction}
\vspace{-1mm}
The advent of autonomous vehicles (AVs)—defined as vehicles capable of navigating and operating without human intervention—holds immense potential to revolutionize the transportation sector \cite{10138317}. These vehicles promise safer roads and more sustainable travel, reducing accidents and emissions. Achieving these benefits requires integrating advanced communication technologies and distributed edge computing to enable real-time responsiveness, decision-making, and enhanced computational power \cite{10518091}. Among these, the emerging sixth-generation (6G) technologies play a critical role in this evolution, empowering AVs with the ultra-reliable low-latency communications (URLLC), thereby creating a safer and more efficient transportation network \cite{9779322}.

As AVs operate in highly dynamic and complex environments, achieving the stringent latency and reliability targets of URLLC demands more than traditional communication techniques. In particular, the use of short packets encoded with finite blocklength (FBL) codes has become indispensable, as it directly addresses the fundamental trade-off between latency, reliability, and throughput in real-world systems \cite{10106132}. Unlike conventional infinite blocklength assumptions, FBL coding allows for realistic performance guarantees in low-latency scenarios by optimizing packet size and error probability jointly. This approach significantly minimizes end-to-end latency while maintaining the high reliability required for real-time decision-making in AVs and other critical applications \cite{10173693}.

The high-mobility context significantly influences wireless communication, as applications like base station (BS) to vehicles, high-velocity rail systems, and low earth orbit satellite networks are notably susceptible to mobility-related challenges. Addressing the detrimental effects of Doppler shift and delay induced by high-mobility is thus of paramount importance \cite{9170653}. The concept of rate-splitting was originally introduced in the 1990s by Rimoldi and Urbanke as a theoretical strategy to improve the efficiency of communication over multi-user channels  \cite{485709}. Building upon this foundation, rate-splitting multiple access (RSMA) has recently emerged as a practical and powerful multiple access scheme tailored to modern wireless systems. It manages multi-user interference by partially treating it as noise, similar to space division multiple access (SDMA), and partially decoding it, as in non-orthogonal multiple access (NOMA) \cite{10400885}. Intrinsically, this is achieved through rate-splitting at the transmitter and successive interference cancellation (SIC) at the receivers \cite{10198464}. By effectively balancing interference and noise, RSMA exhibits strong robustness, making it particularly suitable for high-mobility scenarios \cite{9831440}. Unlike SDMA, which tends to degrade under high user mobility, RSMA maintains reliable multi-user connectivity and significantly enhances throughput \cite{9491092}. Given the demonstrated capabilities of RSMA in challenging high mobility environments, there is a compelling necessity to investigate its application in finite blocklength settings to further address the critical Quality of Service (QoS) demands to latency and reliability \cite{10106132}.

Recently, many studies have explored RSMA and its applications due to the above characteristics. Mao \emph{et al.} \cite{9158344} demonstrated that RSMA could capture a broader rate region compared to dirty paper coding in multi-antenna broadcast channels when only partial channel state information (CSI) at the transmitter (CSIT) is accessible. The work in \cite{10858168} investigates uplink RSMA-based cell-free massive multiple-input multiple-output (mMIMO) systems with low-resolution ADC constraints, providing closed-form expressions for rate and energy efficiency. Further, the study in \cite{10720715} derives closed-form secrecy and legitimate rate expressions under hardware impairments and pilot spoofing, and proposes an SCA-based power control scheme to enhance RSMA security in cell-free mMIMO systems. In a high-mobility system, the authors of \cite{9491092} discussed the infinite blocklength issues of achieving the maximum ergodic sum-rate in an underloaded scenario and provided a closed-form power solution. Additionally, Dizdar \emph{et al.} \cite{10316238} discussed the maximum fairness ergodic rate in an overloaded scenario and proposed a closed-form power solution as well. These findings collectively demonstrate the versatility and effectiveness of RSMA in addressing various communication challenges, suggesting its potential as a viable solution for next-generation wireless networks.


\vspace{-3mm}
\subsection{Related Works}
Due to the stringent latency and reliability requirements of URLLC, recent studies have increasingly focused on incorporating RSMA into short-packet communication systems. In \cite{9831048}, RSMA was applied to a downlink SPC system with perfect CSI, where an optimal linear precoder was designed to maximize the instantaneous rate. The authors of \cite{10251998} extended this idea to terahertz (THz)-based SPC systems with SIC imperfections, targeting extreme data rate and latency demands in 5G. In \cite{10478577}, a closed-form expression for the ergodic sum-rate was derived for a high-mobility SPC system under large-scale fading, addressing URLLC constraints. More recent work in \cite{10535307} investigates the application of RSMA in cell-free massive MIMO systems for URLLC scenarios. It derives analytical rate bounds under imperfect CSI and proposes a geometric programming-based algorithm for power allocation. Additionally, \cite{10813419} applied RSMA to wireless networked control systems, jointly optimizing beamforming, rate allocation, and user pairing to meet communication reliability and control stability requirements.

In most existing works of RSMA under conditions of imperfect CSIT, resource allocation like beamforming and power distribution were optimized at instantaneous moments using metrics including generalized mutual information (GMI) \cite{9893376}, short-term average rate (AR) \cite{7555358},\cite{9620766}, and interrupt probability \cite{10032137}. These optimizations are often addressed through iterative algorithms like interior point methods, which generally involve high complexity and result in delays. To facilitate practical applications, the RSMA framework can employ low-complexity precoders, such as zero forcing (ZF) and maximum ratio transmission (MRT), which are based on instantaneous channel conditions\cite{10720715}, \cite{7152864,8008852,7434643,8393474,8907421}. Additionally, it manages power allocation decisions according to long-term channel characteristics, such as large-scale fading. These aspects are widely explored in the literature \cite{9491092},\cite{10316238},\cite{10535307}, \cite{10000714}. This framework has been expanded to include applications with finite blocklengths as shown in \cite{10478577}, where the optimization of power distribution was performed using a complex exhaustion. 

Nevertheless, these prior implementations of RSMA frameworks, which employ low-complexity linear precoding and long-term channel characteristics, are typically based on two assumptions. \textbf{i)} The first assumption involves the uniform power allocation of private streams among users \cite{9491092},\cite{10316238},\cite{10478577},\cite{10000714},\cite{10173504}. While this strategy of equal power distribution is prevalent in multi-user multiple input multiple output (MU-MIMO) systems, as extensively documented in both academic and practical contexts \cite{8008852}, \cite{6476877}, \cite{7888974}, it may not be the most effective approach for the rate-splitting regime. Unlike SDMA, RSMA can improve system performance with non-uniform power allocation to private streams, while ensuring fairness via common rate splitting. \textbf{ii)} The second assumption adopted is that QoS requirements can be met regardless of the rate allocation strategy, as prior studies often neglect the analysis of common rate splitting and QoS constraints \cite{9491092},\cite{10720715}, \cite{10478577}, \cite{10000714}, \cite{10173504}. In dynamic environments with moving vehicles, RSMA can flexibly allocate parts of the common stream for customized data delivery to each user. This is essential for meeting specific QoS demands. However, when the common stream is deactivated, adjustments in the power distribution among various private streams are often necessary to maintain QoS standards, particularly for users in poor channel conditions.
\vspace{-3mm}
\subsection{Motivations and Contributions}
Autonomous driving requires ultra-reliable and low-latency communication, which necessitates modeling short-packet transmissions using FBL theory. Conventional schemes often suffer performance degradation under high mobility, whereas RSMA demonstrates robustness. Therefore, a low-complexity RSMA-FBL framework holds significant practical value for real-world deployment. As indicated in the above references, the framework that employs linear precoding and power optimization based on large-scale fading is anticipated to achieve lower complexity. Previously, these systems often relied on uniform allocation of private streams, overlooking the flexibility of RSMA to satisfy user QoS requirements through common rate splitting.
In this study, we enable RSMA FBL transmissions to support high-mobility autonomous driving and explicitly evaluate the ergodic performance by considering multiple vital factors.
 
A closed-form expression for the ergodic rate is derived, and power allocation and rate splitting are jointly optimized to meet QoS requirements. 
The main contributions of this paper are summarized as follows.

\begin{itemize}
	\item[$\bullet$] We investigate the use of RSMA in highly mobile systems that have specific FBL and block error rate (BLER) requirements for downlink transmissions. In our approach, common streams are handled using random precoding, while private streams utilize ZF precoding. Instead of adhering to average allocations for private stream power and common rate splitting, the proposed scheme offers flexible adjustments for these parameters.
	\item[$\bullet$] We derive the closed-form lower bound to the ergodic rate of RSMA, considering additional factors such as vehicle positions, the power of each stream, blocklengths, and BLERs. Utilizing the continuous nature of the closed-form expressions, we apply these forms to the optimization process using gradient-based methods. This approach avoids the errors typically introduced by the discontinuity of sample average methods.   
	\item[$\bullet$]
    We formulate an optimization framework for RSMA under imperfect CSIT and QoS constraints. We first optimize the global power coefficient for the common stream. When activated, we sequentially optimize the private power distribution and common rate splitting. If deactivated, we use sequential quadratic programming for private power optimization. This strategy facilitates targeted adjustments based on the status of the common stream, enhancing the adaptability of the system.
	\item[$\bullet$]We validate the tightness of the derived lower bound by comparing it with other closed-form expressions and evaluate the performance of the proposed RSMA scheme through numerical simulations. The proposed scheme is compared with SDMA and RSMA using average private power distribution and common rate splitting. The results show that the proposed RSMA scheme significantly improves both the ergodic sum-rate and minimum vehicle rate, enhancing fairness and sum-rate performance.
\end{itemize}

The rest of the paper is organized as follows. Section~\ref{sec:system} presents the system model. Section~\ref{sec:formulation} provides the closed-form ergodic rate analysis. In Section~\ref{problem}, we formulate the problem of maximizing the ergodic sum-rate and use the closed-form expressions to optimize the allocation of global power, private power, and common rate splitting. Section~\ref{result} presents the numerical results, and Section~\ref{conclusion} concludes the paper.

\textit{Notation:} Vectors are represented by bold lowercase letters. The operators $|\cdot|$ and $\|\cdot\|$ signify the cardinality of a set or the absolute value of a scalar, and the $l_2$-norm of a vector, respectively. $\mathbf{a}^T$ and $\mathbf{a}^H$ denote the transpose and Hermitian transpose of vector $\mathbf{a}$, respectively. The term $\mathcal{CN}(0, \sigma^2)$ describes the circularly symmetric complex gaussian distribution with zero mean and variance $\sigma^2$. ${\mathbf{I}}_N$ stands for an $N \times  N$ identity matrix. The rounding operation is indicated by $\lfloor \cdot \rceil$. Natural logarithms are denoted by $\ln(\cdot) = \log_e(\cdot)$. $Gamma(D, \theta)$ denotes the distribution with the probability density function $f(x) = \frac{x^{D-1} e^{-x/\theta}}{\Gamma(D) \theta^D}$, where $D$ and $\theta$ are the shape and scale parameters, respectively. $\Gamma(D) = \int_0^\infty t^{D-1} e^{-t} \ \mathrm{d}t$ represents the gamma function, and $\Gamma'(D)$ denote its derivative. ${E_v}(x) = \int_1^\infty {{e^{ - tx}}{t^{ - v}} \ {\rm{d}}t}$, with $x > 0$ and $v \in \mathbb{R}$, denotes the generalized exponential integral function. The operator $\mathbb{E}\{\cdot\}$ denotes the expectation over the distribution of random variables.

%
\vspace{-2mm}
\section{System Model}
\label{sec:system}
\vspace{-1mm}
\setlength{\abovecaptionskip}{-0.05cm}
\setlength{\belowcaptionskip}{-0.2cm}
\label{Sec_SysMod}

Fig. 1 illustrates a typical scenario in high-mobility autonomous driving. The BS is equipped with \( N_t \) antennas in a downlink multiple-input single-output (MISO) broadcast setup and serves \( K \) high-mobility single-antenna AVs, where \( N_t \geqslant K \). These vehicles are indexed by \( \mathcal{K} = \{1, 2, \ldots, K\} \). A single-layer RSMA strategy is adopted to reduce receiver complexity \cite{2018Rate}. Each AV is equipped with a single-layer SIC for decoding. The message for each vehicle, denoted as \( W_k \) for \( k \in \mathcal{K} \), is divided into two components: the common part \( W_{c,k} \) and the private part \( W_{p,k} \). At the BS transmitter, all common parts \( \{W_{c,1}, \ldots, W_{c,K}\} \) are collectively encoded into a common stream \( s_c \). This common stream can be decoded by all vehicles and is used to convey multicast messages, such as platoon coordination and model distribution \cite{10032163}. Each private message \( W_{p,k} \) is encoded into a private stream \( s_k \) and would be decoded independently by each vehicle after the common stream is decoded\footnote{RSMA is a relatively recent technique and has not yet been explicitly incorporated into 3GPP standards. However, several of its underlying components, such as MU-MIMO and multicast, have already been partially standardized \cite{10038476}.}. Generally, the private part can be used for unicast information delivery that supports various services, such as path planning and remote driving \cite{10130367}. The composite stream vector to be transmitted is represented as \( \mathbf{s} = [s_c, s_1, s_2, \ldots, s_K]^T \). Linear precoding is applied across all streams. The transmitted signal at the BS is expressed as

\begin{align}
	\label{transmit signal}
	{\mathbf{x}} = \sqrt {{P}{(1-t)}} {{\mathbf{p}}_c}{s_c} + \sqrt{Pt}\sum_{{k \in {\mathcal K}}} {\sqrt {\mu_k} {{\mathbf{p}}_k}{s_k}}, 
\end{align}
where ${{\mathbf{p}}_c} \in {\mathbb{C}^{{N_t \times 1}}}$ and ${{\mathbf{p}}_k} \in {\mathbb{C}^{{N_t \times 1}}}$ are the precoders pertainting to the common stream and the $k$-th private stream, respectively, satisfying ${\left\| {{{\mathbf{p}}_c}} \right\|^2} = {\left\| {{{\mathbf{p}}_k}} \right\|^2} = 1,\forall k\in{\mathcal K}$. Here $P$ represents the total transmit power at the BS, and the composite stream vector ${{\mathbf{s}}} \in \mathbb{C}^{{(K+1)} \times 1}$ ensures \mbox{$\mathbb{E}\left\lbrace \mathbf{s}\mathbf{s}^{H}\right\rbrace =\mathbf{I}$}. The global power coefficient $0 \leq t \leq 1$ governs the division of power between the common and total private streams, where $1-t$ is the power ratio for common stream. The private power distribution $0 \leq \mu_{k} \leq 1$, with $\sum_{k=1}^{K}\mu_k=1,$ specifies the distribution of power across private streams allocated to vehicle-$k$ and define $\mathbf{\boldsymbol\mu}=[\mu_{1}, \mu_{2}, \ldots, \mu_{K}]$. Hence $t\cdot\mu_k$ is the power ratio for the $k$-th private stream. For the sake of reducing complexity, we employ random precoding for the common stream and ZF precoding for individual streams.

\begin{figure}[t]\centering\includegraphics[width=1\linewidth]{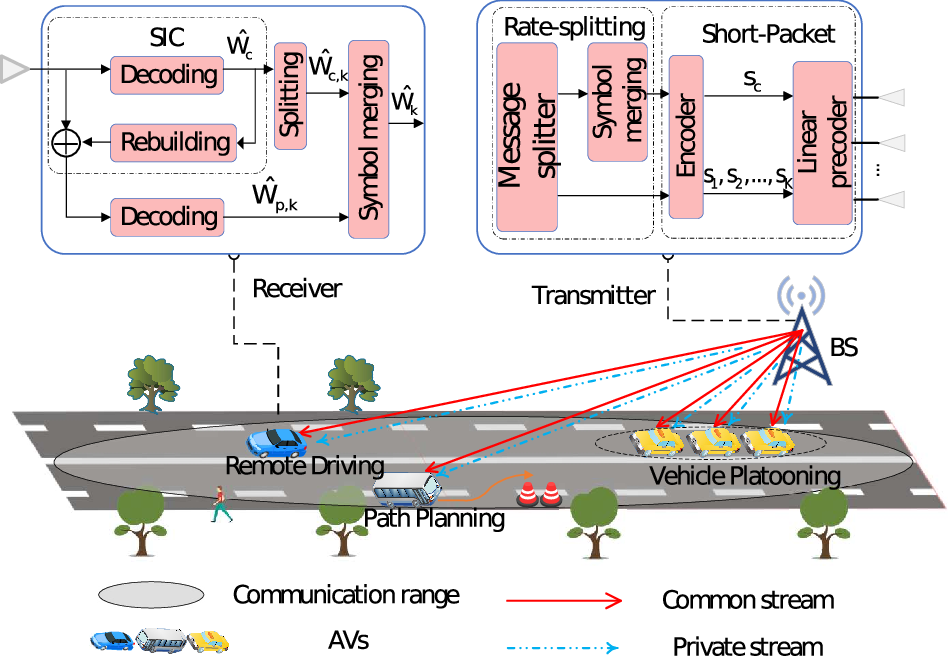}
	\vspace{1.5mm}
	\caption{RSMA-assisted FBL transmissions in a high-mobility autonomous driving scenario.}
\vspace{-2mm}	\label{fig:example}
\end{figure}
The received signal at vehicle-$k$ is given by
\begin{align}
	{y}_{k}=\sqrt{\xi_{k}}\mathbf{h}_{k}^{H}\left[ m \right]{\mathbf{x}}+z_{k}, \quad k\in\mathcal{K},  
\end{align}
where ${{\mathbf{h}}_k}[m] \in {\mathbb{C}^{{N_t} \times 1}}$ represents the small-scale fading at time instant $m$, $z_{k} \sim \mathcal{CN}(0,\sigma^2)$ indicates the additive white gaussian noise (AWGN) at vehicle-$k$, with $\sigma^2$ representing the corresponding noise power. And $\xi_{k}$ represents the average power drop between the BS and vehicle-$k$ due to the large-scale fading. 

To simplify the subsequent analysis, we normalize the noise by scaling the entire signal as
\begin{align}
	\quad y_{k}^{\text{scaled}}=\sqrt{\zeta_{k}}\mathbf{h}_{k}^{H}\left[ m \right]{\mathbf{x}}+n_{k}, \quad k\in\mathcal{K},  
\end{align}
where $\zeta_{k}=\xi_{k}/\sigma^2$, and \( n_{k} \sim \mathcal{CN}(0,1) \) represents the normalized noise.

Initially, each vehicle decodes the common stream \( s_c \), regarding all private streams \{\( s_1,s_2..., s_K \)\} as noise. Thus, the signal-to-interference plus noise ratio (SINR) for the common stream at vehicle-\( k \) can be written as
\begin{align}
	\label{SINR_c}
	{\Gamma _{c,k}} = \frac{{{P}{(1-t)}{\zeta_{k}}{{\left| {{\mathbf{h}}_k^H[m]{{\mathbf{p}}_c}} \right|}^2}}}{{{P}{t}{\zeta_{k}}\sum_{j \in {\mathcal K}} {{\mu_j}{{\left| {{\mathbf{h}}_k^H[m]{{\mathbf{p}}_j}} \right|}^2} + 1} }}. 
\end{align}
Generally, each vehicle applies SIC, decoding the common stream ${s_c}$ and subtracting it from the received signal \( y_k \). To ensure that all users successfully decode ${s_c}$, the SINR of the weakest user, i.e., \(\min_{k \in \mathcal{K}},\Gamma_{c,k}\) should be focused. Since the system targets ultra-reliable communications, we assume perfect SIC, meaning the common stream is decoded and removed without error\footnote{Although the assumption of perfect SIC ignores potential decoding errors in the common stream that could affect private stream decoding, it is a commonly adopted assumption in RSMA literature (e.g.,~\cite{10478577}) to characterize theoretical performance upper bounds.}. Subsequently, vehicle-\( k \) decodes its private stream ${s_k}$, treating other private streams intended as background noise. The SINR for decoding its private stream at vehicle-\( k \) is
\begin{align}
	\label{SINR_p}
	{\Gamma _{p,k}} = \frac{{{P}{t}{\zeta_{k}}{\mu_k}{{\left| {{\mathbf{h}}_k^H[m]{{\mathbf{p}}_k}} \right|}^2}}}{{{P}{t}{\zeta_{k}}\sum\nolimits_{j \in {\mathcal K},j \ne k} {{\mu_j}{{\left| {{\mathbf{h}}_k^H[m]{{\mathbf{p}}_j}} \right|}^2} + 1} }}. 
\end{align}
After decoding both the common and private streams, vehicle-\( k \) synthesizes the full message by integrating the decoded common message \(\widehat W_{c,k} \) with the decoded private part \(\widehat W_{p,k} \). Following \cite{5452208}, which examined finite blocklength systems, we calculate the instantaneous rates \( R_{c}^{(m)} \) for the common stream and \( R_{p,k}^{(m)} \) for the private stream at time instant \( m \) as follows\begin{align}
	\label{R_ck}
{R_{c}^{(m)} } \approx {C} ({\min_{k\in\mathcal{K}}\Gamma_{c,k}}) - \sqrt {\frac{{{V} ({\min_{k\in\mathcal{K}}\Gamma_{c,k}})}}{{{l_c}}}} {{Q} ^{ - 1}}({\beta _{c,k}}), 
\end{align}
\begin{align}
	\label{R_pk}
	{R_{p,k}^{(m)} } \approx {C} ({\Gamma _{p,k}}) - \sqrt {\frac{{{V} ({\Gamma _{p,k}})}}{{{l_k}}}} {{Q} ^{ - 1}}({\beta _{p,k}}), 
\end{align}
where $\{{l_c},{l_1},...,{l_K}\}$ are the respective blocklengths of streams $\{{s_c},{s_1},...,{s_K}\}$. The common transmission rate shown in \eqref{R_ck} ensures that all vehicles can decode the common stream. The parameters $\beta_{c,k}$ and $\beta_{p,k}$ represent the BLERs requirements for vehicle-$k$ to decode the common and private messages, respectively. The function $C(x) = \log_2(1 + x)$ represents the Shannon capacity. The function ${Q}^{-1}(\cdot)$ corresponds to the inverse of the Gaussian $Q$ function and $V(\cdot)$ corresponds to the channel dispersion parameter, with $V(x) = (1 - (1 + x)^{-2})(\log_2e)^2$. 

Therefore, the ergodic rates for both the common and private streams are expressed as \( R_c = \mathbb{E} \left\{ R_c^{(m)} \right\} \), \( R_k = \mathbb{E} \left\{ R_{p,k}^{(m)} \right\} \), where \( \mathbb{E}\{\cdot\} \) denotes the expectation over small-scale fading and other random variables. Then, the ergodic sum-rate is given by
\begin{align}
	\label{eqn:sumrate}
	R_{sum}=R_{c}+\sum_{k=1}^{K}R_{k}.
\end{align}

We assume that both the vehicle mobility and CSI reporting delays lead to imperfect CSIT. Therefore, the channel coefficient of vehicle-$k$ at the time $m$ can be given as \cite{5723047}
\begin{align}
	\label{h_m}
	{{\mathbf{h}}_k}[m] = \sqrt {{\varepsilon ^2}} {{\mathbf{h}}_k}[m - 1] + \sqrt {1 - {\varepsilon ^2}} {{\mathbf{e}}_k}[m], 
\end{align}
where \( \mathbf{h}_k[m] \) is identified as a Rayleigh flat fading channel exhibiting spatial uncorrelation, with each component following a \(\mathcal{CN}(0, 1)\) distribution independently and identically. Similarly, the estimation error \( \mathbf{e}_k[m] \) consists of components also distributed as \(\mathcal{CN}(0, 1)\). The time correlation coefficient is given by $\varepsilon = J_0(2\pi f_{D,k} T), (0 \leq \varepsilon \leq 1)$, where \(J_0(\cdot)\) is the zeroth-order Bessel function. This coefficient quantifies the dependency between \(\mathbf{h}_k[m]\) and \(\mathbf{h}_k[m-1]\) according to the Jakes' model \cite{proakis2008digital}.
 The parameter \( T \) represents the channel coherent interval, while the maximum Doppler frequency \( f_{D,k} \) is determined using $f_{D,k}={v_k f_c}/{c}$, where $ v_{k} $ is the vehicle velocity, \( f_c \) is the carrier frequency, and \( c=3\times10^8\rm{m/s} \) is the speed of light. Additionally, we assume a large-scale fading model \( L(d_k) = \rho_0 - 10\alpha_0\lg(d_k/d_0) \) that remains constant throughout the entire coherence interval, where $\rho_0$ specifies the path loss at the standard distance $d_0$, $d_k$ represents the distance between BS and vehicle-$k$. $\alpha_0$ indicates the path loss exponent. The large-scale fading term $\xi_{k} = 10^{L(d_k)/10}$ adjusts accordingly. And we assume that large-scale fading remains constant throughout the entire transmission period.

\section{Closed-Form Ergodic Rate Analysis}
\label{sec:formulation}
\vspace{-1mm}
This section is dedicated to deriving the closed-form lower-bound for the ergodic rate of transmission streams under arbitrary transmission antenna configurations, BLER, and blocklength. Firstly, the distributions of commonly used random variables are provided. Subsequently, utilizing their approximations, closed-form lower bounds for the rates of common and private streams are further derived for use in the optimization discussed in the following subsection.
\vspace{-3mm}
\subsection{Gamma Distribution for Precoding}
We are using ZF precoders for the private streams. Explicitly, the ZF precoder \(\mathbf{p}_k\) is isotropically distributed and operates independently from the Gaussian noise \(\mathbf{e}_j^H[m]\), for all $j \in \mathcal{K} \backslash k$, resulting in a random variable (r.v.) with exponential distribution and unit mean, hence \(\left|\mathbf{e}_j^H[m]\mathbf{p}_k\right|^2 \sim Gamma(1, 1)\). Noticeably, channel knowledge at the transmitter is confined to the channel vector at the previous time instance \( m-1 \), i.e., \(\mathbf{h}_k[m - 1]\). Consequently, we have \(\left|\mathbf{h}_j^H[m - 1]\mathbf{p}_k\right| = 0\), where $j \in \mathcal{K} \backslash k $, and \(\left|\mathbf{h}_k^H[m - 1]\mathbf{p}_k\right|^2\) follows \(Gamma(N_t - K + 1, 1)\), as outlined in \cite{6967804}. The common precoder \(\mathbf{p}_c\) is assumed to function as a stochastic beamformer unrelated to \(\mathbf{h}_k[m - 1]\), \(\mathbf{e}_k[m]\), ensuring \(\left|\mathbf{h}_k^H[m]\mathbf{p}_c\right|^2 \sim Gamma(1,1)\) \cite{7152864}. For simplification, the expression \(\left|\mathbf{h}_k^H[m]\mathbf{p}_j\right|^2\) is approximated for all \(j, k \in \mathcal{K}\) as
\begin{align}
	\label{Approximation1}
	{\small {\left| {{\mathbf{h}}_k^H[m]{{\mathbf{p}}_k}} \right|^2} \approx {\varepsilon ^2}{\left| {{\mathbf{h}}_k^H[m - 1]{{\mathbf{p}}_k}} \right|^2} + (1 - {\varepsilon ^2}){\left| {{\mathbf{e}}_k^H[m]{{\mathbf{p}}_k}} \right|^2}, }
\end{align}
\begin{equation}
	\label{Approximation2}
	{\small 	\begin{split}
			{\left| {{\mathbf{h}}_k^H[m]{{\mathbf{p}}_j}} \right|^2} &\approx {\varepsilon ^2}{\left| {{\mathbf{h}}_k^H[m - 1]{{\mathbf{p}}_j}} \right|^2} + (1 - {\varepsilon ^2}){\left| {{\mathbf{e}}_k^H[m]{{\mathbf{p}}_j}} \right|^2}\\
			&= (1 - {\varepsilon ^2}){\left| {{\mathbf{e}}_k^H[m]{{\mathbf{p}}_j}} \right|^2}, j\ne k. 
	\end{split}}
\end{equation}
Notice that the approximations made in \eqref{Approximation1} and \eqref{Approximation2} disregard the components involving both ${\mathbf{h}}_k^H[m - 1]$ and ${{\mathbf{e}}_k}[m]$ \cite{6967804}. Since most of the variables here follow a Gamma distribution, we introduce a method \cite{6967804} to approximate the random variable \( Z = \sum_i A_i \), where each \( A_i \) follows a Gamma distribution, i.e., \( A_i \sim {Gamma}(D_{A_i}, \theta_{A_i}) \). An estimation of $Z$ is given by the r.v. $\widetilde Z$, which follows $Gamma(\widetilde D_Z, \widetilde \theta_Z)$, with
\begin{align}
	\label{gammajinsi}
	\widetilde D_Z = \frac{{{{\left( {\sum {{D_{A_i}}} {\theta _{A_i}}} \right)}^2}}}{{\sum {{D_{A_i}}} \theta _{A_i}^2}},\widetilde \theta_{A_i}  = \frac{{\sum {{D_{A_i}}} \theta _{A_i}^2}}{{\sum {{D_{A_i}}} {\theta _{A_i}}}}. 
\end{align}
\vspace{-4mm}
\subsection{Lower Bound for $R_c$}
We first propose the following approximation to obtain a closed-form lower bound for the ergodic rate of the common stream, i.e., $R_c$.
\newtheorem{lemma}{Lemma}
\newtheorem{definition}{Definition}
\begin{definition}
	\label{empirical approximation}
	Define a new random variable \( Y_{c,k} \) as follow
\begin{align}
	\label{Cc_jinsi}
	{Y _{c,k}} = \frac{{{{\left| {{\mathbf{h}}_k^H[m]{{\mathbf{p}}_c}} \right|}^2}}}{{\frac{{P}{t}\zeta_{k}}{K}\sum_{j \in {\mathcal K}} {{{\left| {{\mathbf{h}}_k^H[m]{{\mathbf{p}}_j}} \right|}^2} + 1} }}. 
\end{align} 
Note that \( {P}(1-t)\zeta_{k}Y_{c,k} \) is equivalent to \( \Gamma_{c,k} \) in \eqref{SINR_c}  with \( \mu_j = 1/K \) for all \( j \in \mathcal{K} \), and then we make the following approximations
	\begin{align}
		\label{Cminc}
\mathbb {E}\left \lbrace{C{(\min_{k\in\mathcal{K}}\Gamma_{c,k})}}\right \rbrace \approx \mathbb {E}\left \lbrace{C{({P}{(1-t)}\min_{k\in\mathcal{K}}(\zeta_{k}Y_{c,k}))}}\right \rbrace,
	\end{align}	
		\begin{align}
			\label{vminc}
			\begin{gathered}
			\mathbb {E}\left \lbrace{\sqrt {\frac{{{V}(\min_{k\in\mathcal{K}}\Gamma_{c,k})}}{{{l_c}}}} {{Q} ^{ - 1}}({\beta _{c,k}})}\right \rbrace \hfill \\
				\approx \mathbb {E}\left \lbrace{\sqrt {\frac{{{V} ({P}{(1-t)}{\min_{k\in\mathcal{K}}(\zeta_{k}Y_{c,k}}))}}{{{l_c}}}} {{Q} ^{ - 1}}({\beta _{c,k}})}\right \rbrace.   \hfill \\ 
			\end{gathered}
		\end{align}
This approximation is based on extensive experimental data and will be validated through the simulation results presented in Fig.~\ref{errors}.
\end{definition}

\begin{lemma}
	\label{Xjinsi}
	The r.v. $X = \sum\nolimits_{j \in {\mathcal K}} {{{\left| {{\mathbf{h}}_k^H[m]{{\mathbf{p}}_j}} \right|}^2}} $ can be closely approximated by a r.v. $\widetilde X$ with the distribution $Gamma({\widetilde D},{\widetilde \theta})$, where
\begin{align}
	\widetilde{D}&=\frac{\left[ \varepsilon^{2}(N_{t}+1)+(1-2\varepsilon^{2})K\right] ^{2}}{\varepsilon^{4}(N_{t}+1)+(1-2\varepsilon^{2})K}, \nonumber \\
	\widetilde{\theta}&=\frac{\varepsilon^{4}(N_{t}+1)+(1-2\varepsilon^{2})K}{\varepsilon^{2}(N_{t}+1)+(1-2\varepsilon^{2})K}.
	\label{eqn:dandtheta}
\end{align}
\end{lemma}
\begin{proof}
See Appendix A
\end{proof}
After deriving the distributions of the above r.v.s, we give the approximate cumulative distribution function (CDF) of ${{P}{(1-t)}\min_{k\in\mathcal{K}}Y_{c,k}}$ in the following Lemma. 
\begin{lemma}
		\label{lemma3}
	The r.v. ${P}{(1-t)}\min_{k\in\mathcal{K}}{Y_{c,k}}$ can be closely approximated by a r.v. ${\widetilde \Gamma _{c}}$  with the CDF
	\begin{align}
		\label{CDF_c}
		{{F} _{{{\widetilde \Gamma }_{c}}}}(y) = 1 - \frac{{{e^{ -\sum_{k=1}^{K}\frac{1}{\zeta_{k}} \frac{y}{{{P}{(1-t)}}}}}}}{{{{\left( {y\frac{{{{\widetilde\theta}{t}}}}{{{K}{(1-t)}}} + 1} \right)}^{{{\widetilde D}{K}}}}}}, 
	\end{align}
	where ${y\in\text{[0,}}\infty {\text{)}}$, ${\widetilde D}$ and ${\widetilde \theta}$ are specified in \eqref{eqn:dandtheta}. 
\end{lemma}
\begin{proof}
	See Appendix B
\end{proof}
Getting the CDF of ${\widetilde \Gamma _{c}}$, we give the closed-form approximations of 
$\mathbb {E}\left \lbrace{C{(\widetilde\Gamma_{c})}}\right \rbrace$ and $\mathbb {E}\left \lbrace{{\widetilde\Gamma_{c}}}\right \rbrace$ in the following Lemma.
\begin{lemma}
	\label{lemma4}
  The CDF of the r.v. ${\widetilde \Gamma _{c}}$ is equivalent to
	\begin{align}
		{{F} _{{\widetilde \Gamma _{c}}}}(y) = 1 - \frac{{{e^{ - {C_1}{y}}}}}{{{{\left( {{{{C_2}{y}}} + 1} \right)}^{{{D}}}}}}, 
    \end{align}
	where $C_1={\sum_{k=1}^{K}\frac{1}{\zeta_{k}{P}{(1-t)}}}$,
	$C_2=\frac{{{{\widetilde\theta}{t}}}}{{{K}{(1-t)}}}$, and $D=\left \lfloor{{\widetilde D}{K}}\right\rceil$. The operator \( \left\lfloor \cdot \right\rceil \) denotes rounding to the nearest integer. Then, $\mathbb{E}\left\lbrace{C{(\widetilde\Gamma_{c})}}\right \rbrace$ can be approximated by
	\begin{align}
    \label{E_C_Gamma_c}
		\mathbb{E}\left[ {{{C}}({\widetilde\Gamma_{c}})} \right] \approx \frac{{{e^{\frac{{{C_1}}}{{{C_2}}}}}}}{{\ln 2}}\Psi \left( { {{{D}}}  } \right), 
	\end{align}
where $\Psi (n) = \frac{{{e^{\frac{{{C_1}}}{{{C_2}}}\left( {{C_2} - 1} \right)}}}}{{{{\left( {1 - {C_2}} \right)}^n}}}{E_1}({C_1}) - \sum\nolimits_{i = 1}^n {\frac{{{E_i}\left( {\frac{{{C_1}}}{{{C_2}}}} \right)}}{{{{\left( {1 - {C_2}} \right)}^{n + 1 - i}}}}}$. And $\mathbb{E}\left\lbrace{{\widetilde\Gamma_{c}}}\right \rbrace$ can be equivalent to
	\begin{align}
    \label{E_Gamma_c}
	\mathbb{E}\left[ {{{\widetilde \Gamma }_{c}}} \right] = \frac{{{e^{\frac{{{C_1}}}{{{C_2}}}}}}}{{{C_2}}}\int_1^\infty  {{e^{ - \frac{{{C_1}}}{{{C_2}}}z}}{z^{ - {{D}}}} \ \mathrm{d}z}  = \frac{{{e^{\frac{{{C_1}}}{{{C_2}}}}}}}{{{C_2}}}{E_{{{ D}}}}(\frac{{{C_1}}}{{{C_2}}}), 
\end{align}
where ${E_v}(x) = \int_1^\infty  {{e^{ - tx}}{t^{ - v}}{\rm{d}}t} ,x > 0,v \in \mathbb{R}$ is the generalized exponential-integral.		
\end{lemma}
\begin{proof}
The proof follows from \cite[Lemma 9]{10478577}. 
\end{proof} 

\begin{Theorem}
The lower bound \( \widehat{R}_c(t) \) for the ergodic rate of the common stream \( R_c \) is given by
\begin{equation}
\label{R_c_lower_bound}
\widehat{R}_c(t) = \frac{e^{\frac{C_1}{C_2}}}{\ln 2} \Psi \left( D \right) - \sqrt{\frac{V\left\lbrace \frac{e^{\frac{C_1}{C_2}}}{C_2} E_D\left(\frac{C_1}{C_2}\right) \right\rbrace}{l_c}} Q^{-1}(\beta_{c,k}),
\end{equation}
where the terms \( C_1 \), \( C_2 \), and \( D \) are defined in Lemma \ref{lemma4}.
\end{Theorem}
\begin{proof}
Using \eqref{Cminc} and \eqref{vminc} along with Lemma \ref{lemma3}, we approximate \( R_c \) as
\begin{equation}
R_c \approx \mathbb{E}\left\lbrace C(\widetilde\Gamma_c) \right\rbrace - \mathbb{E}\left\lbrace \sqrt{\frac{V(\widetilde\Gamma_c)}{l_c}} Q^{-1}(\beta_{c,k}) \right\rbrace. \label{Za}
\end{equation}
Since \( \sqrt{V(x)} \) is concave, applying Jensen’s inequality to \eqref{Za}, we obtain the following bound,
\begin{equation}
\label{R_c_lower}
R_c \ge\mathbb{E}\left\lbrace C(\widetilde\Gamma_c) \right\rbrace - \sqrt{\frac{V\left\lbrace \mathbb{E}(\widetilde\Gamma_c) \right\rbrace}{l_c}} Q^{-1}(\beta_{c,k}).
\end{equation} 
Substituting \eqref{E_C_Gamma_c} and \eqref{E_Gamma_c} into \eqref{R_c_lower}, we obtain the lower bound of \( R_c \), \( \widehat{R}_c(t) \), thus completing the proof.
\end{proof}
\vspace{-3mm}
\subsection{Lower Bound for $R_k$}
We derive the closed-form lower bound for the ergodic rate of the $k$-th private stream, i.e., $R_k$. Firstly, we rewrite $\mathbb {E}\left \lbrace{C(\Gamma_{p,k})}\right \rbrace$ as
 \begin{small}
\begin{align}
	\label{log2}
	\mathbb {E}&\left \lbrace{C(\Gamma_{p,k})}\right \rbrace=\mathbb{E}\left\lbrace \log_{2}\left( 1+ {Pt}{\zeta_{k}}\sum_{j\in\mathcal{K}}{\mu_j}|\mathbf{h}^{H}_{k}[m]\mathbf{p}_{j}|^{2}  \right) \right\rbrace\nonumber\\&-\mathbb{E}\left\lbrace\log_{2}\left(1+ {Pt}{\zeta_{k}}\sum_{j\in\mathcal{K}, j\neq k} \mu_j|\mathbf{h}^{H}_{k}[m]\mathbf{p}_{j}|^{2} \right) \right\rbrace.  
\end{align}
\end{small}To obtain a concise closed-form approximation of $\mathbb {E}\left \lbrace{C(\Gamma_{p,k})}\right \rbrace$, we propose the following Lemmas.
\begin{lemma}
	\label{Xkdejinsi}
	The r.v. $X_{k} = \sum\nolimits_{j \in {\mathcal K}} {{\mu_j}{{\left| {{\mathbf{h}}_k^H[m]{{\mathbf{p}}_j}} \right|}^2}}$ can be approximated by the r.v. $\widetilde X_k$ with the distribution $Gamma({\widetilde D_{X_k}},{\widetilde \theta _{X_k}})$, where
	\begin{align}
		\label{dxkthxk}
		{\widetilde D_{X_k}} = \frac{(({N_t-K+1}){\varepsilon^2}{\mu_{k}}+1-{\varepsilon^2})^2}{(N_t-K+1)\varepsilon^4\mu_{k}^2+{(1-\varepsilon^2)^2}{\sum\nolimits_{j \in {\mathcal K}} {{\mu_j}^2} }},\hfill\\{\widetilde\theta_{X_k}}=\frac{{(N_t-K+1)\varepsilon^4\mu_{k}^2+{(1-\varepsilon^2)^2}{\sum\nolimits_{j \in {\mathcal K}} {{\mu_j}^2} }}}{({N_t-K+1}){\varepsilon^2}{\mu_{k}}+1-{\varepsilon^2}},
	\end{align}
    and $\alpha_k\triangleq\mathbb{E}\{\ln(\widetilde{X}_k)\}=\ln(\widetilde\theta_{X_k})+\Gamma^{\prime}(\widetilde{D}_{X_k})/\Gamma(\widetilde{D}_{X_k})$.
\end{lemma}
\begin{lemma}
	\label{Ykdejinsi}
	The r.v. $ Y_k= \sum\nolimits_{j \in {\mathcal K}\backslash k} {{\mu_j}{{\left| {{\mathbf{h}}_k^H[m]{{\mathbf{p}}_j}} \right|}^2}} $ can be approximated by the r.v. $\widetilde Y_k$ with the distribution $Gamma({\widetilde D_{Y_k}},{\widetilde \theta _{Y_k}})$, where
\begin{align}
\label{d4o4}
{\widetilde D_{Y_k}} = \frac{{{{(1 - {\mu_k})}^2}}}{{\sum\nolimits_{j \in {\mathcal K}\backslash k} {{\mu_j}^2} }},\ {\widetilde \theta _{Y_k}} = \frac{{(1-\varepsilon^2)}{\sum\nolimits_{j \in {\mathcal K}\backslash k} {{\mu_j}^2} }}{{1 - {\mu_k}}},
\end{align}
    and $\eta_k\triangleq\mathbb{E}(\widetilde{Y}_k)=(1-\varepsilon^2)(1-\mu_{k})$. 
\begin{proof}The proof of Lemmas \ref{Xkdejinsi} and \ref{Ykdejinsi} is similar to Lemma \ref{Xjinsi}, which is omitted here due to page limitations.\end{proof}
\end{lemma}
We similarly give the following definition before obtaining a closed-form lower bound for $R_k$.
\begin{definition}
	\label{lemma5}
	Define a new r.v. $Y_{p,k}$ as follow
	\begin{align}
		\label{Cp_jinsi}
		{Y_{p,k}} = \frac{{{\zeta_{k}}{P}{t}{\mu_{k}}{{\left| {{\mathbf{h}}_k^H[m]{{\mathbf{p}}_k}} \right|}^2}}}{{\frac{{P}{t}\zeta_{k}{(1-\mu_{k})}}{(K-1)}\sum_{j \in {\mathcal K},j \ne k} {{{\left| {{\mathbf{h}}_k^H[m]{{\mathbf{p}}_j}} \right|}^2} + 1} }},
	\end{align} Note that $Y_{p,k}$ is equivalent to \( \Gamma_{p,k} \) in \eqref{SINR_p} with \( \mu_j =\frac{1-\mu_{k}}{K-1}  \) for $\forall j \in \mathcal{K} \backslash k$, and then we make the following approximation
	\begin{align}
		\label{vminp}
		\begin{gathered}
			\mathbb {E}\left \lbrace{\sqrt {\frac{{{V}(\Gamma_{p,k})}}{{{l_k}}}} {{Q} ^{ - 1}}({\beta _{p,k}})}\right \rbrace 
			\approx\mathbb {E}\left \lbrace{\sqrt {\frac{{{V} ({Y_{p,k}})}}{{{l_k}}}} {{Q} ^{ - 1}}({\beta _{p,k}})}\right \rbrace.    \hfill \\ 
		\end{gathered}
	\end{align}
This approximation is based on extensive experimental data, which will also be validated in the simulation results shown in Fig.~\ref{errors}.
 	
\end{definition}
Different from Definition \ref{empirical approximation}, the approximation similar to \eqref{Cminc} is no longer reliable. We will provide this lower approximation in \eqref{R_k_bound}. And then we give the approximated mean of ${Y_{p,k}}$. 

\begin{lemma}
\label{Y_pkdejinsi}
The expectation of \( Y_{p,k} \) can be approximated by the expectation of \( \widetilde{\Gamma}_{p,k} \), i.e., 
\(\mathbb{E}[Y_{p,k}] \approx \mathbb{E}[\widetilde{\Gamma}_{p,k}]\). The expectation \(\mathbb{E}[\widetilde{\Gamma}_{p,k}]\) is given by
\begin{align}
\label{mean_Gamma_p_k}
\mathbb{E}\left[ {{{\widetilde \Gamma }_{p,k}}} \right] = \frac{{{e^{\frac{{{C_{k_1}}}}{{{C_{k_2}}}}}}}}{{{C_{k_2}}}}{E_{{{[K-1]}}}}(\frac{{{C_{k_1}}}}{{{C_{k_2}}}}),
\end{align}where $C_{k_1} = \frac{1}{{P t \zeta_k \mu_k \widetilde{D}_{M_k} \widetilde{\theta}_{M_k}}}$, $C_{k_2} = \frac{{\left( 1 - \mu_k \right)\left( 1 - \varepsilon^2 \right)}}{{\left( K - 1 \right)\mu_k \widetilde{D}_{M_k} \widetilde{\theta}_{M_k}}}$, and $\widetilde{D}_{M_k} = \frac{{\left[ \left( N_t - K + 1 \right)\varepsilon^2 + \left( 1 - \varepsilon^2 \right) \right]^2}}{{\left( N_t - K + 1 \right)\varepsilon^4 + \left( 1 - \varepsilon^2 \right)^2}}$, $\widetilde{\theta}_{M_k} = \frac{{\left( N_t - K + 1 \right)\varepsilon^4 + \left( 1 - \varepsilon^2 \right)^2}}{{\left( N_t - K + 1 \right)\varepsilon^2 + \left( 1 - \varepsilon^2 \right)}}$.

\end{lemma}
\begin{proof}
			See Appendix C
\end{proof}

\begin{Theorem}
The lower bound \( \widehat{R}_k \) for the ergodic rate of the $k$-th private stream \( R_k \) is given by
\begin{align}
\label{R_k_bound}
\nonumber \widehat{R}_k(t,\boldsymbol{\mu})&= \log_{2}\left( 1+ {Pt}{\zeta_{k}}e^{\alpha_k} \right)-\log_{2}\left( 1+ {Pt}{\zeta_{k}}{\eta_k} \right)\\&-{\sqrt {\frac{{{V}\left \lbrace\frac{{{e^{\frac{C_{k_1}}{C_{k_2}}}}}}{{{C_{k_2}}}}{E_{{{[K-1]}}}}(\frac{{{C_{k_1}}}}{{{C_{k_2}}}}).\right \rbrace}}{{{l_k}}}} {{Q} ^{-1}}({\beta _{p,k}})}.	
\end{align}
where the terms $\alpha_k$, $\eta_k$, \( C_{k_1} \) and \( C_{k_2} \) are defined in Lemmas \ref{Xkdejinsi}, \ref{Ykdejinsi} and \ref{Y_pkdejinsi}.
\end{Theorem}
\begin{proof}
We first apply the linearity of expectation to rewrite \( R_k \) as
\begin{align}
\nonumber R_k &= \mathbb{E}\left\lbrace \log_2\left( 1+ {Pt}{\zeta_{k}} X_k \right) \right\rbrace 
- \mathbb{E}\left\lbrace \log_2\left( 1+ {Pt}{\zeta_{k}} Y_k \right) \right\rbrace \\
& \quad - \mathbb{E}\left\lbrace \sqrt{\frac{V(\Gamma_{p,k})}{l_k}} Q^{-1}(\beta_{p,k}) \right\rbrace.
\end{align}
 Using Lemmas \ref{Xkdejinsi}, \ref{Ykdejinsi}, and \eqref{vminp}, we approximate $R_k$ as
\begin{align}
\nonumber R_k &\approx \mathbb{E}\left\lbrace \log_2\left( 1+ {Pt}{\zeta_{k}} \widetilde{X}_k \right) \right\rbrace 
- \mathbb{E}\left\lbrace \log_2\left( 1+ {Pt}{\zeta_{k}} \widetilde{Y}_k \right) \right\rbrace \\
&\label{R_k_jeason} \quad - \mathbb{E}\left\lbrace \sqrt{\frac{V(Y_{p,k})}{l_k}} Q^{-1}(\beta_{p,k}) \right\rbrace.
\end{align}
Then, applying Jensen's inequality for the convex functions \( \log_2(1 + ae^x) \), \( -\log_2(1 + ax) \), and \( -a\sqrt{V(x)} \) in \eqref{R_k_jeason}, we derive the following lower bound
\begin{align}
\nonumber
R_k &\ge \log_2\left( 1 + {Pt}{\zeta_{k}} e^{\mathbb{E}\left\lbrace \ln(\widetilde{X}_k) \right\rbrace} \right) 
- \log_2\left( 1 + {Pt}{\zeta_{k}} \mathbb{E}\left\lbrace \widetilde{Y}_k \right\rbrace \right) \\
& \label{E_Y_pk_approx_E_Gamma_pk} \quad - \sqrt{\frac{V\left\lbrace \mathbb{E}(Y_{p,k}) \right\rbrace}{l_k}} Q^{-1}(\beta_{p,k}).
\end{align}
Based on Lemma \ref{Y_pkdejinsi}, we can make \(\mathbb{E}[Y_{p,k}] \approx \mathbb{E}[\widetilde{\Gamma}_{p,k}]\) in \eqref{E_Y_pk_approx_E_Gamma_pk}. Finally, by applying Lemmas \ref{Xkdejinsi}, \ref{Ykdejinsi}, and \ref{Y_pkdejinsi} to perform variable substitution for \( \mathbb{E}\left\lbrace \ln(\widetilde{X}_k) \right\rbrace \), \( \mathbb{E}\left\lbrace \widetilde{Y}_k \right\rbrace \), and \( \mathbb{E}[\widetilde{\Gamma}_{p,k}] \), we obtain the final lower bound \( \widehat{R}_k \).
\end{proof}

\vspace{-2mm}
\section{Problem Formulation and Solution}
\label{problem}
This section strives to enhance the ergodic performance in RSMA FBL transmissions with the obtained bounds. To maximize the ergodic sum-rate, we
jointly optimize the global power coefficient, private power
coefficient and common rate splitting. Explicitly, the problem is formulated as follows
\begin{subequations}\label{eqn:problem_p0}
\begin{align}
	\text{\(\mathcal{P}_0\)}: \max_{ t, \boldsymbol{\mu}, \mathbf{c}}\   &  \sum_{k=1}^{K}(C_k+{R}_{k})  \label{p0_obj} \\
	s.t. \quad & 0\leq t \leq 1,  \label{p0_t_0_1} \\
	& \sum_{k=1}^{K}\mu_k=1,  \label{p0_muk_sum} \\
	& 0\leq \mu_{k} \leq 1, \forall k \in {\mathcal K},  \label{p0_muk_0_1} \\
	& C_k+R_k \geq R_{\min}, \forall k \in {\mathcal K}, \label{p0_Ck_min} \\
	& \sum_{k\in\mathcal{K}}C_{k} \leq R_{c}, \forall k \in {\mathcal K}, \label{p0_Ck_sum} \\
	& C_k \geq 0, \forall k \in {\mathcal K}  \label{p0_Ck_0}
\end{align}
\end{subequations}
where $C_{k}$ represents the portion of the common stream rate allocated to user-$k$, and $R_{\min}$ indicates the QoS constraints for each vehicle. The optimization variable $\mathbf{c}=[C_{1}, C_{2}, \ldots, C_{K}]$ determines the rate allocation among users for the common stream. The variable $t$, global power coefficient, represents the power ratio between the common and the total private streams, where $1-t$ specifies the power ratio for the common stream. Besides, $\mathbf{\boldsymbol\mu} = [\mu_1, \mu_2, \ldots, \mu_K]$ specifies the power ratio distribution among private streams, hence $t \cdot \mu_k$ designates the power ratio for the $k$-th private stream. Constraint \eqref{p0_t_0_1} concerns power allocation for the common stream. Constraints \eqref{p0_muk_sum} and \eqref{p0_muk_0_1} correspond to the power allocation among private streams. Constraint \eqref{p0_Ck_min} specifies the minimum transmission rate constraints for receivers. Constraint \eqref{p0_Ck_sum} ensures that all receivers reliably decode the common stream, and \eqref{p0_Ck_0} indicates the splitted common rate is non-negative.

It is evident that the interdependence of variables and their ratios within the objective function makes problem ($\mathcal{P}_0$) inherently non-convex. Solving this problem without a global search presents substantial challenges, as there is no closed-form expression for \eqref{p0_obj}. Moreover, applying exhaustive search methods directly is associated with prohibitive computational complexity. To find a more efficient solution with low complexity, we use the derived bounds $\widehat{R}_c$ in \eqref{R_c_lower_bound} and $\widehat{R}_k$ in \eqref{R_k_bound} to approximate the ergodic sum-rate performance. 

Specifically, we first optimize the global power coefficient for equal power-allocated private streams. Then, we decide whether to enable the common stream based on previous results. If the common stream is enabled, we will first ignore the common rate splitting, optimize private power allocation, and then optimize the common rate splitting based on all allocation ratios. If the common stream is not enabled, we will introduce QoS constraints \eqref{p0_Ck_min} and proceed with problem-solving. This entire optimization process is referred to as the single-step optimization algorithm. Finally, we will briefly discuss complexity.
\vspace{-3mm}
\subsection{Global Power Allocation}
We temporarily disregard the common rate splitting and keep the private power distribution constant. Therefore, we omit the private allocation constraints \eqref{p0_muk_sum}, \eqref{p0_muk_0_1}, and rate allocation constraints \eqref{p0_Ck_min}, \eqref{p0_Ck_sum}, \eqref{p0_Ck_0}. By employing the method of contradiction, it becomes evident that when problem $\mathcal{P}_0$ is maximized, constraint $\eqref{p0_Ck_sum}$ holds at equality; otherwise, it is clear that a higher ${C_k}$ can be chosen, thereby achieving a higher sum-rate. 
Then, we substitute equation $\eqref{p0_Ck_sum}$, where equality holds, into the objective function and use the closed-form expression. This leads to the following global power allocation subproblem
\begin{subequations}
\begin{align}
\label{p1_obj}
{\mathcal{P}_1}:~\max_{t}\   &   \widehat{R}_c(t)+\sum_{k=1}^{K}\widehat {R}_{k}(t,\boldsymbol{\mu})\hfil\\ s.t.\ \label{p1_t_0_1}&0\leq t \leq 1,	
\end{align}
\end{subequations}
where $\widehat{R}_c(t)$ is shown in \eqref{R_c_lower_bound} and $\widehat{R}_k(t,\boldsymbol{\mu})$ is calculated by \eqref{R_k_bound} with
$\boldsymbol{\mu}=[1/K,1/K...,1/K]\in {\mathbb{C}^{{K}\times 1}}$ fixed. Even though there is only one variable and one linear constraint, the objective function is non-elementary. Using the exhaustive method, the calculation can take an extremely long time. 

Since the objective function in \eqref{p1_obj} has continuous first-order derivatives, we use a gradient-based iterative algorithm that employs sequential quadratic programming (SQP) methods for rapid convergence\footnote{Although SQP methods for non-convex problems can be sensitive to initialization and may converge to local optima, our simulation results show that the proposed algorithm generally exhibits satisfactory robustness and feasibility.}. At each major iteration, a quasi-Newton update is used to approximate the Hessian of the Lagrangian function. This approximation is then used to generate a quadratic programming (QP) subproblem, whose solution provides the search direction for the line search procedure. The step length is determined by an appropriate line search, and the entire procedure iterates until convergence \cite{articleconstrained}. The initial value of \( t \) is set to 0.5. This choice ensures a balanced power split between common and private streams, with equal power initially allocated to each user. According to the optimized \( t^* \), we perform subsequent subproblems based on whether the common stream is activated.


\vspace{-3mm} \subsection{Private Power Allocation and Rate-splitting}
\subsubsection{$t^*\ll1$}
As $t$ specifies the power ratio for total private streams, in the case of $t^*\ll1$, most of the power will be allocated to the common stream. In such conditions, we can assume that the common stream has a sufficiently large rate available for satisfying the QoS requirement in \eqref{p0_Ck_min}. Hence, we disregard common rate splitting to optimize the private power allocation $\boldsymbol\mu$. Observing \eqref{Cc_jinsi} and \eqref{Cminc}, the common stream rate does not undergo significant changes due to variations in the private power allocation. Therefore, we ignore the common stream and consider the following problem\begin{subequations}
\begin{align}
\label{p2_obj}
{\mathcal{P}_2}: ~\max_{\boldsymbol{\mu}}\ &   \sum_{k=1}^{K}\widehat{R}_{k}(t^*,\boldsymbol{\mu})\hfil\\ s.t. &\sum_{k=1}^{K}\mu_k=1,\hfil\\&0\leq \mu_{k} \leq 1,\forall k \in {\mathcal K}
\end{align}
\end{subequations}
The objective function of ${\mathcal{P}_2}$ is highly coupled non-convex. Achieving the globally optimal solution also requires a very high level of complexity. Therefore, we still use SQP to obtain a suboptimal solution. The initial value of \(\boldsymbol{\mu}\) is set to \(\mu_k = {\zeta_{k}}/{(\sum_{i=1}^{n} \zeta_i)}, \forall k \in \mathcal{K}\). Setting the initial point in this manner enables faster convergence and helps to avoid premature convergence, compared to starting from the average private stream allocation.

\begin{remark}
In fact, allocating lower power for the private stream to users with poorer channels can enhance the overall rate of the common stream. This occurs because weaker users, when decoding the common stream, no longer need to treat their high-power private stream signals as interference.
\end{remark}

After solving $\mathcal{P}_1$ and $\mathcal{P}_2$, the power allocation for the initial problem $\mathcal{P}_0$ has been determined. Once the power allocation is specified, the actual $R_c$ and $R_k$ can be computed using Monte Carlo simulations. Next, we only need to configure the common rate splitting to satisfy the QoS requirements in \eqref{p0_Ck_min}. Since all users are assumed to have the same minimum rate requirement, the common stream data can be allocated to the user with the lowest rate to meet these requirements. This results in a rate allocation problem, which can be formulated as maximizing the rate of the user with the minimum rate. The optimization problem in $\mathcal{P}_0$ can then be reformulated as
\begin{subequations}
\begin{align}
\label{p3_obj}
	\mathcal{P}_3:~&\max_{\boldsymbol{c}}\min_{\forall k \in {\mathit{K}}}   (C_k + {R}_{k}) \\
	s.t. &\sum_{k\in\mathcal{K}}C_{k} = {R}_{c}, \forall k \in {\mathcal K} \\
	&C_k \geq 0, \forall k \in {\mathcal K}
\end{align}
\end{subequations}
Given that all powers are already specified, the corresponding achievable ergodic rates for each stream are fixed. Although closed-form expressions could be used to reduce the complexity, using the lower bound to allocate the rate may introduce errors. To avoid the effect of these errors, we estimate $R_c$ and $R_k$ using sample average approximation (SAA)
\begin{align}  
	\label{R_ck_meng}
	{R}_{c} &=\lim _{M \rightarrow \infty } \frac {1}{M} \sum _{m=1}^{M} R_{c, k}^{(m)}, \\ \label{R_pk_meng} {R}_{k} &=\lim _{M \rightarrow \infty } \frac {1}{M} \sum _{m=1}^{M} R_{p, k}^{(m)}. 
\end{align}
where $M$ is the number of Monte Carlo trials. As the number of sampling points approaches infinity, the SAA converges to the expected value  \cite{7555358}. After substituting a larger \( M \) to estimate \( R_c \) and \( R_k \), this problem transforms into a water-filling problem. The final level is calculated by $R_{level}=({R}_c+\sum_{i=1}^{j}{R}_k^{(i)})/j$, where $j$ starts from $K$ and decreases to 1. If ${R}_k^{(j)} \leq R_{level}$, the loop exits. Here, ${R}_k^{(j)}$ represents the $j$-th smallest private rate among ${R}_{k}$. Finally, the result is:
\begin{align}
	{C}_k= 
	\begin{cases}
		R_{level}-{R}_k & \text{if } 0\leq R_{level}-{R}_k \\
		0 & \text{if }  R_{level}-{R}_k < 0
	\end{cases}
\end{align}

\vspace{-2mm}
\subsubsection{$t^*\to1$} 
In practical systems, this leads to closing common streams, causing the RSMA scheme to degrade into an SDMA one. Due to the closure of common streams, rate splitting no longer holds significance, i.e., $C_k=0, \forall k \in {\mathcal K} $. If the SDMA scheme does not consider QoS constraints, it tends to allocate all power to users with the best channel conditions, thereby failing to meet min-rate constraint \eqref{p0_Ck_min} for all vehicles. Next, we consider the following problem
\begin{subequations}
\begin{align}
	\label{p4_obj}
{\mathcal{P}_4}:~\max_{\boldsymbol{\mu}}\   &   \sum_{k=1}^{K}{\widehat{R}}_{k}(t^*,\boldsymbol{\mu})\hfil\\ s.t \ &\sum_{k=1}^{K}\mu_k=1,\hfil\\&0\leq \mu_{k} \leq 1,\forall k \in {\mathcal K}\hfil\\
	& \widehat R_k \geq R_{min}, \forall k \in {\mathcal K} 
\end{align}
\end{subequations}
This is similar to $\mathcal{P}_2$, with the additional introduction of the minimum rate constraint. Since the closed-form solution is non-convex and involves non-elementary functions, we do not have widely applicable solving methods. Here, we continue to use the SQP method to solve $\mathcal{P}_4$. The initial value of $\boldsymbol{\mu}$ is set to $\mu_k=1/K,\forall k \in {\mathcal K}$. Although this initial point may not fully satisfy all QoS constraints, the SQP algorithm generally possesses constraint correction capabilities that can help guide the solution towards the feasible region during iterations, resulting in a reliable solution \cite{Nocedal2006}.
\begin{remark}
The value of $t^*$ is significant in these two cases. Findings in \cite{9491092} show that the optimal $t^*$ is often much less than or equal to 1. This result aligns with literature on rate-splitting degrees of freedom \cite{9257433}\cite{8019852}, indicating that maximizing degrees of freedom requires more power for the common stream than for private streams.
\end{remark}
In this paper, we focus on improving the ergodic sum-rate performance. In situations where $t^*=0.5$, the rates of opening and closing the common stream are often close to each other. Moreover, the solution to $\mathcal{P}_1$ naturally tends to polarize toward the extremes (i.e., $t^* \approx 0$ or $t^* \approx 1$). Therefore, in our optimization settings, we treat \( t^* \leq 0.5 \) as \( t^* \ll 1 \). For \( t^* > 0.5 \), we consider \( t^* \) to be close to 1 and set \( t^* = 1 \) to terminate the common stream. While this simplification may overlook partially active common stream cases, simulations show that the proposed method still improves overall performance in most scenarios. The overall algorithm is sketched in Algorithm 1.

\begin{algorithm}[t]
\caption{The Proposed Single-Step Update Algorithm}
\label{Algorithm_P2}
\begin{algorithmic}[1]
\State Initialize $\boldsymbol{\mu}=[1/K,...,1/K]$, calculate the $t$ by solving $\mathcal{P}_1$ with SQP
\If{$t^* > 0.5$}
    \State Reset $t^* = 1$
    \State Calculate the $\boldsymbol{\mu}$ by solving $\mathcal{P}_4$ with SQP
\Else
   \State Calculate the $\boldsymbol{\mu}$ by solving $\mathcal{P}_2$ with SQP
        \State Given $t^*$, $\boldsymbol{\mu}$, calculate the $\boldsymbol{c}$ by solving $\mathcal{P}_3$ with water-filling
\EndIf
\end{algorithmic}
\end{algorithm}
\vspace{-3mm}
\subsection{Computational complexity}
In each iteration, the algorithm updates variables based on the gradient of the objective function. The expression in \eqref{R_c_lower_bound}, derived from Lemma~\ref{lemma4}, has a complexity of \( \mathcal{O}(K^2) \), where \( K \) denotes the number of vehicles. The expression in \eqref{R_k_bound} is computed \( K \) times, resulting in a total complexity of \( \mathcal{O}(K) \). Since $\mathcal{P}_3$ has a closed-form solution and its computation time is shorter compared to the others, its complexity can be disregarded. 

All other subproblems are solved via the SQP method. The complexity of each iteration depends on the total number of variables and constraints, and generally grows cubically with their sum. \( \mathcal{P}_1 \), which must be solved in both cases, involves a one-dimensional variable and a single constraint, yielding a complexity of \( \mathcal{O}(N_1 K^2) \), where \( N_1 \) is the number of iterations. Notably, this problem’s complexity is limited by the closed-form expansion. Based on its solution, either \( \mathcal{P}_2 \) or \( \mathcal{P}_4 \) is subsequently solved. Both problems involve \( K \)-dimensional optimization variables. \( \mathcal{P}_2 \) has \( 1 + 2K \) constraints and a complexity of \( \mathcal{O}(27 N_2 K^3) \), where \( N_2 \) denotes the number of iterations. \( \mathcal{P}_4 \) has \( 1 + 3K \) constraints and a complexity of \( \mathcal{O}(64 N_4 K^3) \), where \( N_4 \) is the number of iterations. Therefore, the total complexity is approximated as \( \mathcal{O}(N_1 K^2 + 27 N_2 K^3) \) or \( \mathcal{O}(N_1 K^2 + 64 N_4 K^3) \), depending on whether the common stream is enabled.

%
%

\vspace{-2mm}
\section{Simulation Results}
\label{result}
In this section, simulations are carried out to evaluate the effectiveness of our derived bounds and the enhancement of the ergodic performance. Unless otherwise specified, all configurations are as follows: the BS is configured with eight transmit antennas (\(N_t = 8\)) and serves four single-antenna vehicles (\(K = 4\)). The total transmit power $P$ is 35 dBm. The time correlation coefficient is $\varepsilon=0.5$, corresponding to vehicle speeds around $v_k=110~\rm{km/h}$. The carrier frequency is \(f_c = 5.9 \, \text{GHz}\) \cite{FCC2020}, with a bandwidth of \( B=10 \, \text{MHz}\) and a noise power density of \(\sigma_0^2 = -174 \, \text{dBm/Hz}\). The noise power is $\sigma^2 = \sigma_0^2 B$. The CSI acquisition delay is \(T = 0.4 \, \text{ms}\), the blocklength is \(l_c = l_k = 300\), and the BLER is \(\beta_{c,k} = \beta_{p,k} = 10^{-6}\) for all streams \cite{9390169}. These settings are suitable for typical vehicle applications, such as safety communication and real-time traffic information exchange \cite{9001049}. We utilize the large-scale fading model \(L(d) = \rho_0 - 10\alpha_0\lg(d/d_0)\), where \(\rho_0 = -30 \, \text{dB}\) specifies the path loss at the standard distance \(d_0 = 1 \, \text{meter}\), \(d\) represents the distance between two terminals, and \(\alpha_0 = 3.7\) indicates the path loss exponent. The minimum required QoS rate \(R_{\text{min}}\) is set at \(0.1 \, \text{bit/s/Hz}\). The positions of vehicles relative to the BS are uniformly distributed at distances between \(100\) and \(300\) meters. All simulation results were generated using a Monte Carlo method based on $10^4$ channel realizations.

\begin{figure}[t]
	\centering
	\begin{subfigure}{0.45\textwidth} 
		\centering
		\includegraphics[width=\linewidth]{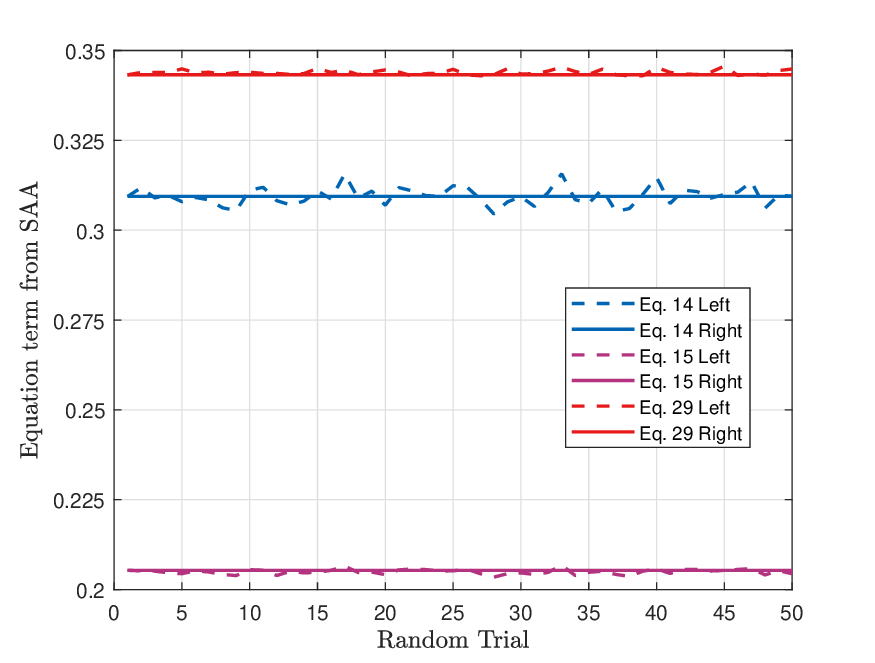} 
         \vspace{-5mm}
		\caption{SAA Computed Values for Approximate Equations.}
		\label{error1}
		\vspace{3mm}
	\end{subfigure}
	\begin{subfigure}{0.45\textwidth} 
		\centering
		\includegraphics[width=\linewidth]{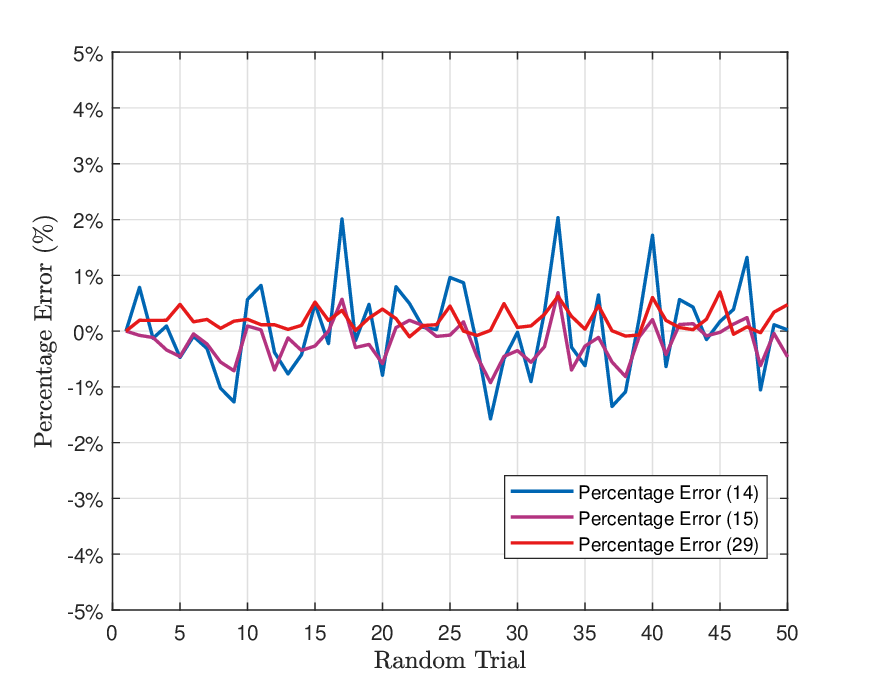} 
        \vspace{-5mm}
		\caption{Percentage Error Calculation.}
		\label{error2}
		\vspace{5mm}
	\end{subfigure}
	\caption{Approximate Data for Definitions 1 and 2.}
	\label{errors}
    \vspace{-3mm}
\end{figure}

Fig. \ref{error1} and~\ref{error2} show the experimental approximations for Definitions~1 and~2. We set $t=0.5$, $v_k = 90~\rm{km/h}$, $\zeta_{k}=1000, \forall k \in {\mathcal K} $. In Fig. \ref{error1}, the left-hand side terms of equations \eqref{Cminc}, \eqref{vminc} and \eqref{vminp} are computed from a set of $\boldsymbol{\mu}$ values generated through random dispersion, where the values satisfy the conditions $u_k = \frac{1}{K}$ and $\sum_{j \in {\mathcal K},j \ne k} u_j = 1 - \frac{1}{K}$ in each random trial. This reflects real power settings. The right-hand side terms of \eqref{Cminc}, \eqref{vminc} and \eqref{vminp} are calculated using $u_k = \frac{1}{K}$ for $\forall k \in \mathcal{K}$. 
From Fig. \ref{error1}, we can see that the left-hand and right-hand terms are essentially on the same level, validating the rationality of the approximations in Definitions 1 and 2. We further illustrate the percentage error in Fig. \ref{error2}. It can be observed that, in 50 random trials, the percentage errors for all three approximations are below 2.5\%. This further confirms the reliability of the proposed approximations.

\begin{figure}[t]
	\centering
	\begin{subfigure}{0.45\textwidth} 
		\centering
		     \vspace{4mm}
             \includegraphics[width=0.95\linewidth]{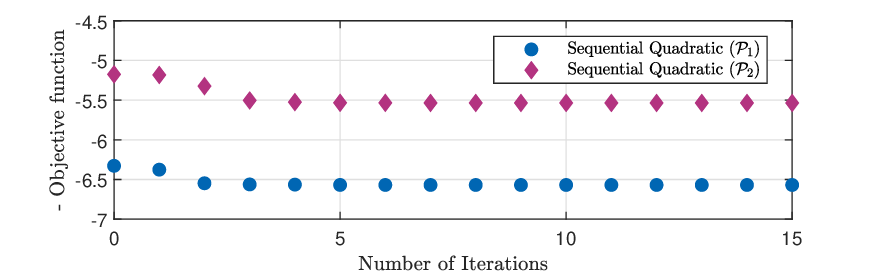} 
        \vspace{-1mm}
		\caption{$P = 35$ dBm, $\varepsilon = 0.8$.}
		\label{fig1a}
		\vspace{0.5cm}
	\end{subfigure}
	\begin{subfigure}{0.45\textwidth} 
		\centering
		\includegraphics[width=0.95\linewidth]{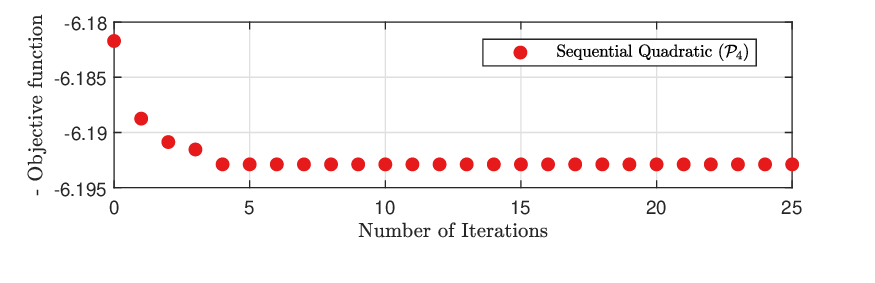} 
        \vspace{-4mm}
		\caption{$P = 30$ dBm, $\varepsilon = 0.8$.}
		\label{fig1b}
		\vspace{0.5cm}
	\end{subfigure}
	\caption{Convergence of the proposed algorithms.}
	\label{fig1}
    \vspace{-2mm}
\end{figure}
In Fig.~\ref{fig1a} and \ref{fig1b}, we illustrate the convergence of the iterative algorithm used to solve problems $\mathcal{P}_1$, $\mathcal{P}_2$, and $\mathcal{P}_4$. Since the default algorithm is primarily for solving minimization problems, we minimize the negative of the objective function to achieve maximization of the original problem. From the figure, it is evident that whether the common stream is activated to handle $\mathcal{P}_1$ and $\mathcal{P}_2$, or deactivated to handle $\mathcal{P}_4$, the solution algorithm can converge to a local minimum value with typically around 5 iterations. This demonstrates that the algorithm converges as the number of iterations increases.

\begin{figure}[t]
	\label{error}
	\centering
	\begin{subfigure}{0.45\textwidth} 
		\centering
		\includegraphics[width=\linewidth]{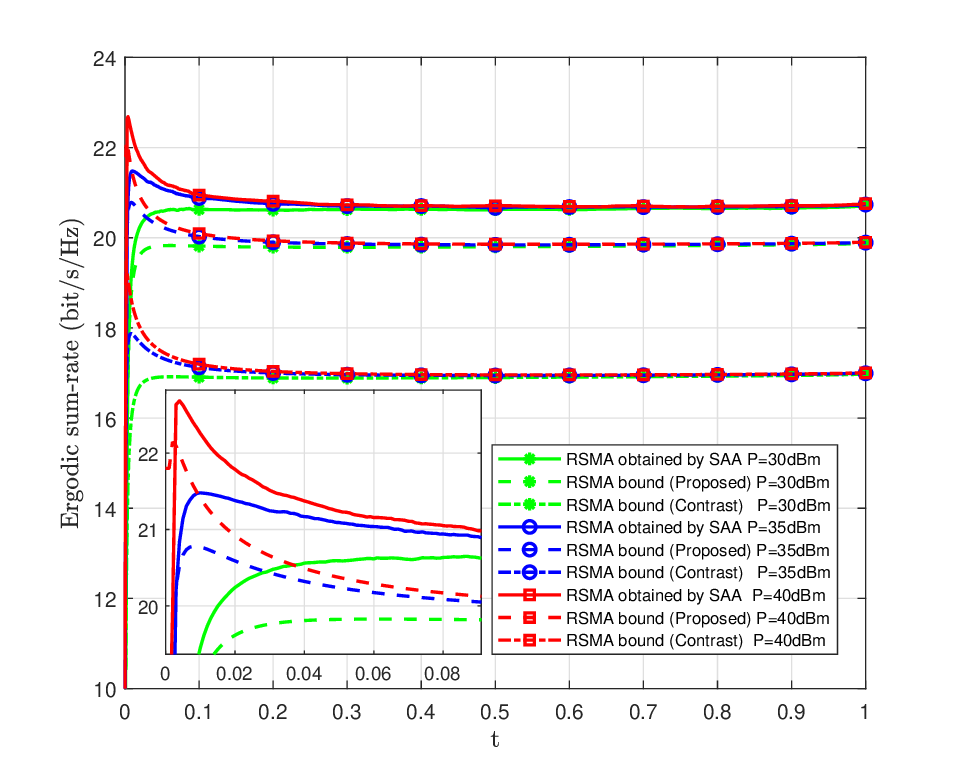} 
        \vspace{-6mm}
		\caption{$N_t$=32, $K$=8, $\varepsilon$=0.8.}
		\label{fig2a}
		\vspace{3mm}
	\end{subfigure}
	\begin{subfigure}{0.45\textwidth} 
		\centering
		\includegraphics[width=\linewidth]{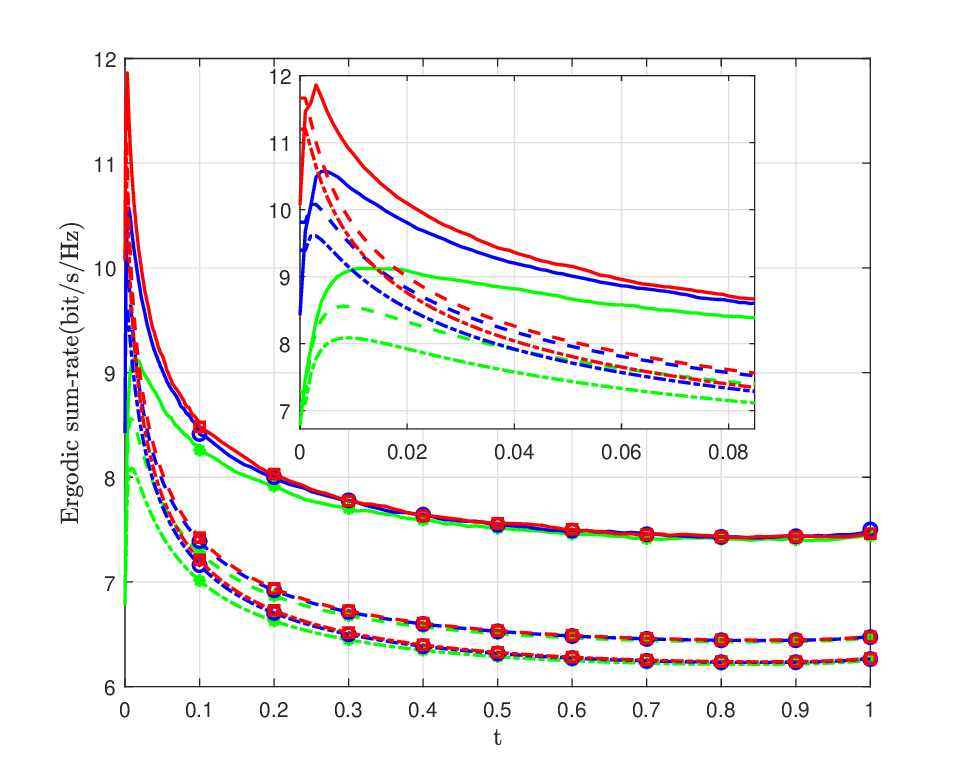} 
          \vspace{-6mm}
		\caption{$N_t$=8, $K$=4, $\varepsilon$=0.8.}
		\label{fig2b}
		\vspace{5mm}
	\end{subfigure}
	\caption{SAA, Proposed and contrast bound.}
	\label{fig2}
    \vspace{-2mm}
\end{figure}
Then, we examine the performance of the derived lower bound via comparisons. Specifically, we compare it with the simulated ergodic sum-rate via SAA with $10^4$ channel realizations and the latest RSMA closed-form bound under the FBL regime provided by \cite{10478577}. This particular bound provides a different closed-form expression for the private stream. However, it does not extend to scenarios involving large-scale fading random distribution, and some points are non-differentiable. In the context of the same large-scale loss of vehicles, the optimal allocation strategy for private streams is equal theoretically. So we set $\zeta_{k}=1000$, $\mu_k=1/K$, $\forall k \in {\mathcal K}$. Then, we can calculate the optimal ergodic sum-rate by exhausting $t$ to demonstrate the optimal performance achieved by RSMA visually. 

In Fig.~\ref{fig2a} and \ref{fig2b}, we display the ergodic sum-rate \eqref{eqn:sumrate} calculated by SAA, the derived lower bound calculated by \eqref{R_c_lower_bound} and \eqref{R_k_bound}, and the closed-form contrast bound in \cite{10478577}. The results show that the proposed lower bound shows a superior approximation of the actual ergodic sum-rate. Notably, the RSMA bound (Proposed) is tighter in 32-antenna systems and consistently outperforms the RSMA bound (Contrast) even in the 8-antenna case. We denote the global power coefficient \( t \), which maximizes the ergodic sum-rate, as \( t_{\text{opt}} \). As the transmit power increases, both 32-antenna and 8-antenna BS exhibit a trend where \( t_{\text{opt}} \) decreases. This decrease in \( t_{\text{opt}} \) leads to an increase in the proportion allocated to common streams for interference mitigation. Under identical mobility conditions and total transmit power, the \( t_{\text{opt}} \) for the 32-antenna BS is relatively larger, as the additional antennas offer greater degrees of freedom to mitigate interference. Additionally, we observe that when the ergodic sum-rate reaches its maximum, the optimal global power coefficient \( t_{\text{opt}} \) typically approaches 0 or 1.
This finding supports the rationale for problem-solving in Section IV-B, which provides segmenting discussions based on the activation status of the common streams.

In the following, we examine the ergodic performance enhancement with the obtained bounds and the proposed algorithm. Specifically, we compare five schemes:  
\begin{itemize}
    \item Rate-splitting with optimized power and common rate splitting, where $t$, $\boldsymbol{\mu}$ and $\mathbf{c}$ are optimized using the proposed single-step update algorithm, termed RSMA Proposed.
    \item Rate-splitting with equal private power distribution and common rate splitting, where \(t^*\) is obtained by solving $\mathcal{P}_1$, termed RSMA Proposed Equal.
    \item Rate-splitting with equal private power distribution and common rate splitting, but \(t^*\) is determined through an exhaustive search to maximize the ergodic sum-rate within the interval (0,1] with a granularity of 0.001, termed RSMA Exhaustive Equal. 
    \item  Space division multiple access with equal power allocation, termed SDMA.
    \item Space division multiple access with exhaustive private power allocation over \(\mu_k \in [0,1]\) with granularity 0.025, termed SDMA Exhaustive.
    \item Non-orthogonal multiple access based on inter-group spatial division multiplexing. In this scheme, all users apply one layer of SIC and are paired according to their proximity, while inter-group power is averaged. For intra-group power, we utilize a power distribution factor derived from fractional transmit power allocation \cite[Eq. 9]{7511620}, setting this parameter to 0.8, termed NOMA.
    \item Non-orthogonal multiple access based on inter-group spatial division multiplexing, with exhaustive intra- and inter-group power allocation (granularity 0.025). All users apply one layer of SIC and are paired by proximity, termed NOMA Exhaustive.
\end{itemize}

To better illustrate the performance of RSMA, we utilize the left singular vector corresponding to the largest eigenvalue of the channel matrix $\mathbf{H}_1 = [\mathbf{h}_1[m-1], \mathbf{h}_2[m-1], \ldots, \mathbf{h}_K[m-1]]$. This left singular vector is then applied to the common stream in all RSMA schemes.
\begin{figure}[t]
	\centering
	\begin{subfigure}{0.45\textwidth} 
		\centering
		\includegraphics[width=\linewidth]{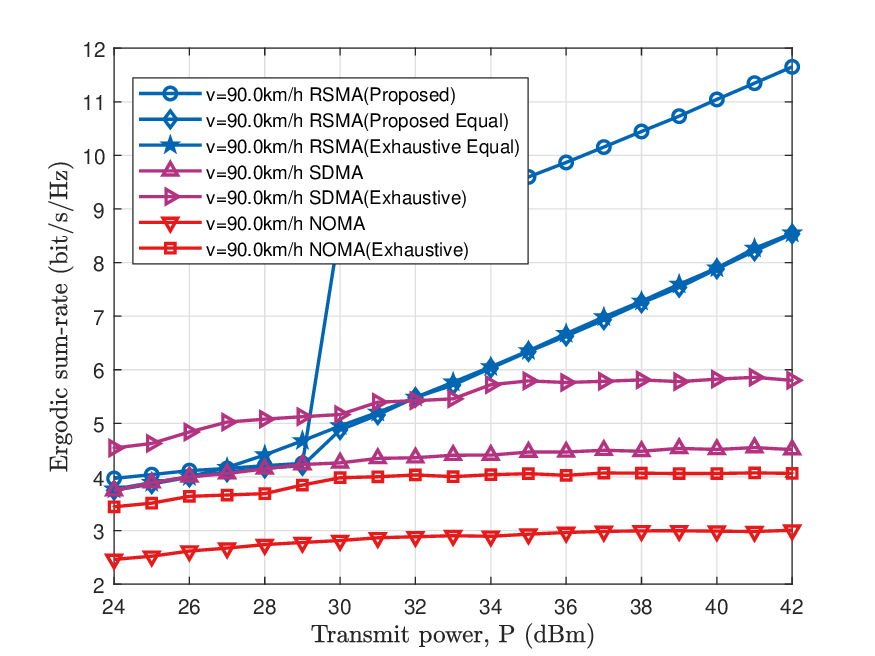} 
        \vspace{-6mm}
		\caption{Vehicle velocity $v_k = 90~\rm{km/h}$.}
		\label{fig3a}
		\vspace{3mm}
	\end{subfigure}
	\begin{subfigure}{0.45\textwidth} 
		\centering
		\includegraphics[width=\linewidth]{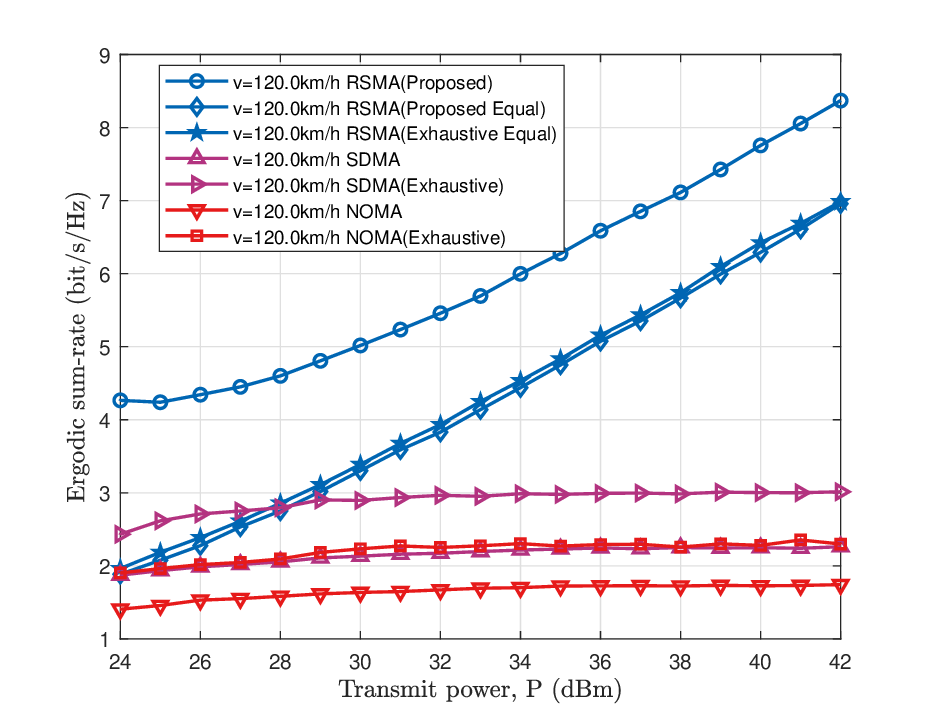} 
         \vspace{-6mm}
		\caption{Vehicle  velocity $v_k = 120~\rm{km/h}$.}
		\label{fig3b}
		\vspace{5mm}
	\end{subfigure}
	\caption{Ergodic sum-rate versus transmit power under different vehicle velocities.}
     \vspace{-2mm}
	\label{Fig3}
\end{figure}

Fig. \ref{fig3a} and \ref{fig3b} illustrate the trends of ergodic sum-rate versus transmit power. It is noteworthy that all RSMA schemes outperform both SDMA and NOMA in the high transmit power regime. This advantage stems from RSMA's inherent ability to effectively manage imperfect channel state information, while SDMA and NOMA are limited in achieving perfect spatial division multiplexing due to imperfect CSI. In contrast, RSMA balances interference and noise through common stream rate splitting and decoding, demonstrating its robustness. Comparing Fig.~\ref{fig3a} and \ref{fig3b}, it can be observed that the ergodic sum-rate of all schemes at 120 km/h is lower than that at 90 km/h. This is due to the increased channel estimation errors at higher speeds. However, RSMA (Proposed), through additional optimization of the private power distributions and the common rate splitting, significantly outperforms other RSMA schemes under average settings. Additionally, RSMA (Proposed Equal), based on solving $\mathcal{P}_1$ with lower bound, and RSMA (Exhaustive Equal), using Monte Carlo methods, demonstrate similar performance under average private streams. This validates the accuracy of the derived closed-form expression. In Fig. \ref{fig3a}, it can be observed that the RSMA (Proposed) achieves a significant performance gain under the condition of 30 dBm. This is because, at this point, the value of \( t^* \) obtained from solving $\mathcal{P}_1$ is less than 0.5. This indicates that the benefit of allocating power to the common stream surpasses that of the private streams, leading to the activation of the common stream. A detailed discussion of the impact of vehicle speeds and transmit powers on the global power coefficient is provided in the discussion of Fig.~\ref{Fig5}. 
\begin{figure}[t]
	\centering
	\begin{subfigure}{0.45\textwidth} 
		\centering
		\includegraphics[width=\linewidth]{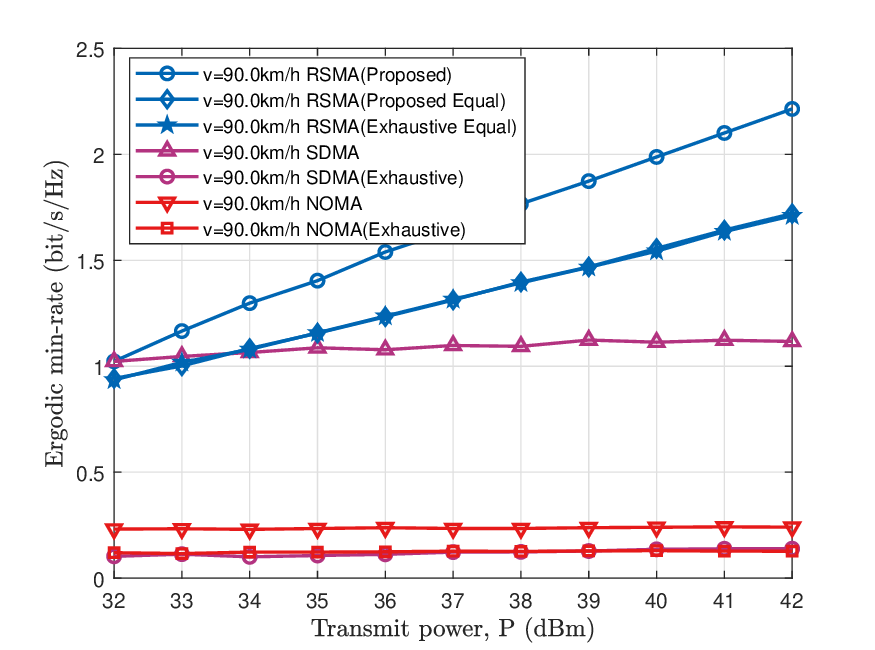} 
         \vspace{-6mm}
		\caption{Vehicle  velocity $v_k = 90~\rm{km/h}$.}
		\label{fig4a}
		\vspace{3mm}
	\end{subfigure}
	\begin{subfigure}{0.45\textwidth} 
		\centering
		\includegraphics[width=\linewidth]{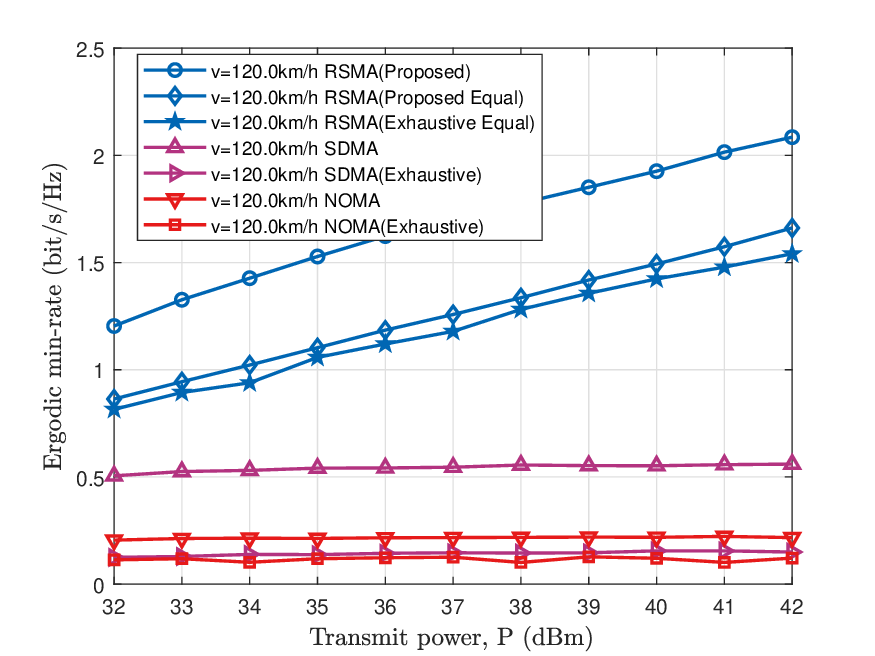} 
          \vspace{-6mm}
		\caption{Vehicle  velocity $v_k = 120~\rm{km/h}$.}
		\label{fig4b}
		\vspace{5mm}
	\end{subfigure}
	\caption{Ergodic min-rate versus transmit power.}
	\label{Fig4}
      \vspace{-2mm}
\end{figure}

In Fig. \ref{fig4a} and \ref{fig4b}, we can observe the ergodic min-rate versus transmit power under the same conditions. It is evident that the proposed scheme surpasses the RSMA with average common stream rate. This is because, based on the optimized \( t^* \), we further jointly optimize the private power distribution and the common rate splitting. By combing Fig. \ref{Fig3} and \ref{Fig4}, it is evident that the proposed scheme can enhance ergodic sum-rate performance by allocating additional power to strong users for the private streams, while also ensuring fairness through a more equitable distribution of the common stream rate. In contrast, SDMA (Exhaustive) and NOMA (Exhaustive) reduce the min-rate to near the QoS threshold in pursuit of higher sum-rate. Consequently, their min-rate performance falls below that of the baseline SDMA and NOMA schemes with fixed parameters.

\begin{figure}[t]\centerline{\includegraphics[width=0.45\textwidth,keepaspectratio]{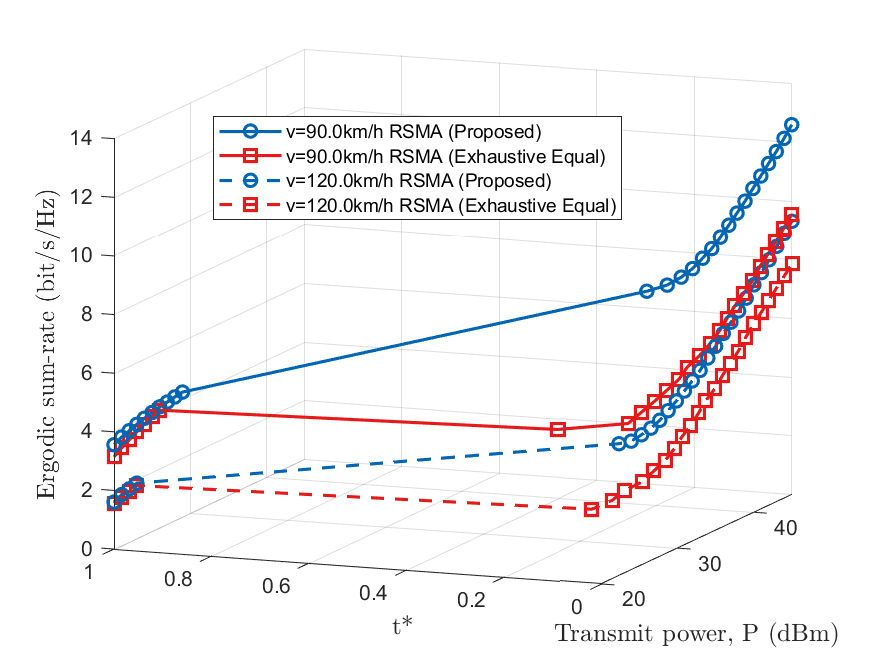}}
	\caption {Ergodic sum-rate versus optimized power ratio $t^*$ and transmit power $P$ under different vehicle velocities.}
	\label{Fig5}
      \vspace{-2mm}
\end{figure}

In Fig. \ref{Fig5}, we display the ergodic sum-rate versus the optimized global power coefficient $t^*$ and the transmit power $P$. Explicitly, $1-t^*$ specifies the optimized power ratio for the common stream, which significantly influences the overall performance. It can be observed that at relatively low transmit power levels, the scheme deactivates the common stream. Once the common stream is activated, as the transmit power increases, the proportion of power allocated to the common stream grows. Additionally, at a vehicle speed of 120 km/h, more prominent interference causes the common stream to be activated at a lower transmit power level. The increase in both transmit power and mobility intensifies the interference among different private streams. Due to imperfect CSIT, zero-forcing cannot eliminate this interference. As a result, allocating power to the common stream provides significant performance gains. Moreover, the $t^*$ value of the RSMA (Proposed) closely aligns with that of RSMA (Exhaustive Equal). This observation confirms the reliability of the derived closed-form solution and demonstrates that the algorithm used to solve $\mathcal{P}_1$ can accurately determine the global power coefficient.

 \begin{figure}[t]
	\centerline{\includegraphics[width=0.45\textwidth,keepaspectratio]{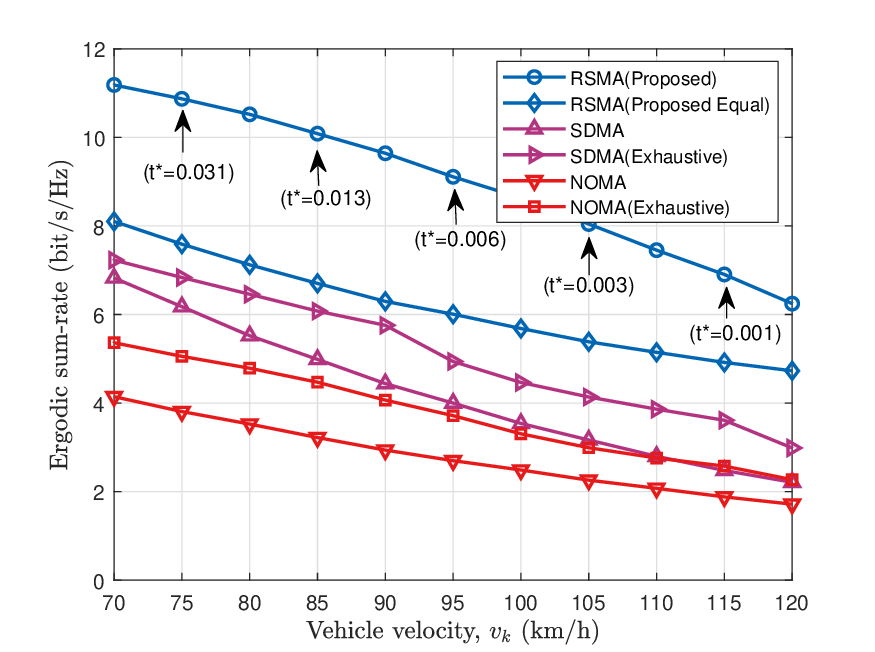}}
	\caption {Ergodic sum-rate versus vehicle velocity.}
	\label{Fig6}
      \vspace{-2mm}
\end{figure}

Fig. \ref{Fig6} depicts the ergodic sum-rate versus vehicle velocity under different schemes. It can be observed that, though the increased velocity degrades the performance overall, RSMA (Proposed) maintains its superior gain in the high-mobility scenario. This demonstrates that the proposed scheme can improve the sum rate, even at higher vehicle velocities, ensuring robust performance under high-mobility conditions. Besides, it is noticed in Fig. \ref{Fig6} that a lower optimized global power coefficient $t^*$ is associated with a higher vehicle velocity $v_k$ in the RSMA (Proposed) scheme. As vehicle velocity increases, the channel estimation error increases, exacerbating inter-user interference between private streams. To mitigate this interference, RSMA tends to allocate more power to the common stream, leading to an increase in $1-t^*$. However, the increase in common stream power also implies reduced power to private streams. This reduction diminishes the performance gain from optimizing private stream power allocation, thereby narrowing the gap between the two RSMA schemes. In contrast, both SDMA and NOMA schemes cannot eliminate interference between vehicles and inter-group interference, resulting in poor performance.

 \begin{figure}[t]
	\centerline{\includegraphics[width=0.45\textwidth,keepaspectratio]{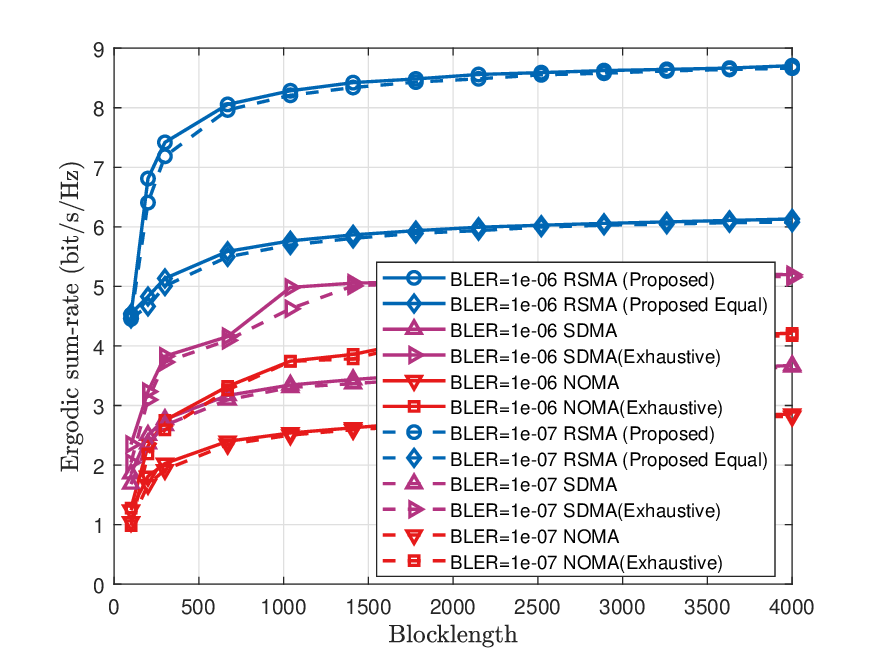}}
	\caption {Ergodic sum-rate versus blocklength under different BLERs.}
	\label{Fig9}
      \vspace{-2mm}
\end{figure}


Fig. \ref{Fig9} displays the performance of the proposed solution under varying blocklengths and BLERs (\(10^{-6}\) and \(10^{-7}\)). As the blocklength increases, performance improves for a fixed BLER, and the proposed scheme consistently maintains the best performance. This indicates that the closed-form solution can effectively monitor the impact of actual FBL. Additionally, the proposed scheme can achieve shorter blocklengths and lower bit error rates without compromising the transmission rate.
  \vspace{-2mm}
\section{Conclusion}
\label{conclusion}

In this work, we have investigated the ergodic sum-rate performance of RSMA in high-mobility autonomous driving scenarios under FBL constraints. We have derived a closed-form lower bound expression considering key system parameters such as the number of transmit antennas, vehicle velocities, power allocation, blocklength, and error rates. Leveraging this expression, we have proposed a gradient-based optimization framework for power allocation and rate splitting, complemented by SQP to ensure QoS and BLER requirements are met. Simulation results have demonstrated that the proposed RSMA scheme significantly improves the ergodic sum-rate while reducing blocklength and error probability, guaranteeing reliable communication even for the worst-case users.

Future work will focus on improving the transmitter framework in two key directions. One is to develop more robust and generalized optimization algorithms to enhance performance stability and solution quality. The other is to investigate multi-layer RSMA designs under more practical condistions such as varying user mobility patterns, spatially correlated fading and realistic channel maps, which can further enhance the robustness compared to single-layer schemes.


\appendices
  \vspace{-2mm}
	\section{PROOF OF LEMMA 1}
By substituting \eqref{Approximation1} and \eqref{Approximation2} into the expression of \(X\), we obtain the following
\begin{equation}
\small
\sum_{j\in\mathcal{K}} {{{\left| {{\mathbf{h}}_k^H[m]{{\mathbf{p}}_j}} \right|}^2}}\approx{\varepsilon ^2}{\left| {{\mathbf{h}}_k^H[m - 1]{{\mathbf{p}}_k}} \right|^2} + (1 - {\varepsilon ^2})\sum_{j\in\mathcal{K}}{\left| {{\mathbf{e}}_k^H[m]{{\mathbf{p}}_j}} \right|^2}
\end{equation}
Since \(\left|\mathbf{h}_k^H[m - 1]\mathbf{p}_k\right|^2\sim Gamma(N_t - K + 1, 1)\) and \(\left|\mathbf{e}_j^H[m]\mathbf{p}_k\right|^2 \sim Gamma(1, 1)\), we can further derive $\varepsilon^2|\mathbf{h}_k^H[m-1]\mathbf{p}_k|^2 \sim {Gamma}(N_t - K + 1, \varepsilon^2)$ and $(1 - \varepsilon^2)\sum_{j\in\mathcal{K}}|\mathbf{e}_k^H[m]\mathbf{p}_j|^2 \sim {Gamma}(K, 1 - \varepsilon^2)$. Substituting these two terms back into the \eqref{gammajinsi}, we obtain \eqref{eqn:dandtheta}.
  \vspace{-2mm}
\section{PROOF OF LEMMA 2}
We initiate the proof by deriving the CDF of the r.v. \(Y_{c,k}\). Firstly, we reformulate \(Y_{c,k}\) into
\begin{align}
	Y_{c,k}=\frac{Q_{k1}}{1+\frac{P\widetilde{\theta}t{\zeta_{k}}}{K}Q_{k2}},
\end{align}
where \(Q_{k1}=|\mathbf{h}^{H}_{k}[m]\mathbf{p}_{c}|^{2}\) and
\(Q_{k2}=\frac{1}{\widetilde{\theta}}\sum_{j\in\mathcal{K}} |\mathbf{h}^{H}_{k}[m]\mathbf{p}_{j}|^{2}\).
As mentioned before, \(Q_{k1}=|\mathbf{h}^{H}_{k}[m]\mathbf{p}_{c}|^{2}\sim {Gamma}(1,1)\). According to Lemma \ref{lemma3}, the r.v. \(X=\sum_{j\in\mathcal{K}} |\mathbf{h}^{H}_{k}[m]\mathbf{p}_{j}|^{2}\) can be approximated by \(\widetilde{X} \sim {Gamma}(\widetilde{D},\widetilde{\theta})\) with \(\widetilde{D}\) and \(\widetilde{\theta}\) provided in \eqref{eqn:dandtheta}. Using this approximation, \(Q_{k2}\) can be represented by \(\widetilde{Q}_{k2} \sim {Gamma}(\widetilde{D},1)\). Assuming the independence of the numerator and the denominator terms \cite{9491092}, we can approximate \(Y_{c,k}\) with \(\widetilde{Y}_{c,k}\), whose CDF is given by 
\begin{small}
	\begin{align}
		F_{\widetilde{Y}_{c,k}}&(y)=\int_{0}^{\infty}P\left(Q_{k1}<y\left( 1+\frac{P\widetilde{\theta}t{\zeta_{k}}}{K}q\right) \right) f_{\widetilde{Q}_{k2}}(q) \ \mathrm{d}q \nonumber\\
		&=\int_{0}^{\infty}\left(1-e^{-y\left(1+ \frac{P\widetilde{\theta}t{\zeta_{k}}}{K}q\right)} \right) \frac{e^{-q}q^{\widetilde{D}-1}}{\Gamma(\widetilde{D})} \ \mathrm{d}q \nonumber\\
		&=\int_{0}^{\infty}\frac{e^{-q}q^{\widetilde{D}-1}}{\Gamma(\widetilde{D})} \ \mathrm{d}q -\frac{e^{-y}}{\Gamma(\widetilde{D})}\int_{0}^{\infty}e^{-q\left(y\frac{P\widetilde{\theta}t{\zeta_{k}}}{K}+1\right)}q^{\widetilde{D}-1} \ \mathrm{d}q \nonumber\\
		&=1 -\frac{e^{-y}}{\Gamma(\widetilde{D})}\int_{0}^{\infty}e^{-q\left(y\frac{P\widetilde{\theta}t{\zeta_{k}}}{K}+1\right)}q^{\widetilde{D}-1} \ \mathrm{d}q.
		\label{eqn:cdf_1}
	\end{align} 
\end{small}Applying a variable substitution $z=q\left(y\frac{P\widetilde{\theta}t{\zeta_{k}}}{K}+1 \right)$, \eqref{eqn:cdf_1} can be rewritten as,
\begin{align}
	&F_{\widetilde{Y}_{c,k}}(y) \nonumber \\
	&=1-\frac{e^{-y}}{\Gamma(\widetilde{D})}\int_{0}^{\infty}e^{-z}\frac{z^{\widetilde{D}-1}}{\left(y\frac{P\widetilde{\theta}t{\zeta_{k}}}{K}+1 \right)^{(\widetilde{D}-1)} }\frac{ \ \mathrm{d}z}{\left(y\frac{P\widetilde{\theta}t{\zeta_{k}}}{K}+1 \right)} \nonumber\\
	&=1 -\frac{e^{-y}}{\left(y\frac{P\widetilde{\theta}t{\zeta_{k}}}{K}+1 \right)^{\widetilde{D}}}\int_{0}^{\infty}\frac{e^{-z}z^{\widetilde{D}-1}}{\Gamma(\widetilde{D})} \ \mathrm{d}z \nonumber \\
	&=1 -\frac{e^{-y}}{\left(y\frac{P\widetilde{\theta}t{\zeta_{k}}}{K}+1 \right)^{\widetilde{D}}}, 
\end{align} 
which is valid for $y \in [0,\infty)$. Consequently, the CDF of $P(1-t)\zeta_{k}\widetilde{Y}_{c,k}$ can be obatined as
\begin{align}
F_{P(1-t)\zeta_{k}\widetilde{Y}_{c,k}}(y)=1 -\frac{e^{-\frac{y}{P(1-t){\zeta_{k}}}}}{\left(y\frac{\widetilde{\theta}t}{{(1-t)}K}+1 \right)^{\widetilde{D}}}.
\end{align}Assuming $Y_{c,k}$ are independent for each user, we can approximate $P(1-t)\min_{k\in\mathcal{K}}(\zeta_{k}Y_{c,k})$ by $\widetilde \Gamma_{c}$ with the CDF $F_{\widetilde\Gamma_{c}}(y)=1-\prod_{k=1}^{K} (1-F_{{P(1-t)\zeta_{k}\widetilde{Y}_{c,k}}}(y))$ and get \eqref{CDF_c}.

\vspace{-2mm}
\section{PROOF OF LEMMA6}
We reformulate \( Y_{p,k} \) in \eqref{Cp_jinsi} as 
\[
Y_{p,k} = \frac{P t \zeta_{k} \mu_k M_k}{P t \zeta_{k} N_k + 1},
\]
where $M_k = {\left| {{\mathbf{h}}_k^H[m]{{\mathbf{p}}_k}} \right|^2}$ and $N_k = \frac{1 - \mu_k}{K - 1} \sum\nolimits_{j \in {\mathcal K}\backslash k} {{{\left| {{\mathbf{h}}_k^H[m]{{\mathbf{p}}_j}} \right|}^2}}$. 
Assuming that the numerator and the denominator are independent, we have
\begin{align}
	\mathbb{E}&\left[ {{{Y }_{p,k}}} \right]=\mathbb{E}\left[ {{Pt\zeta_{k}{\mu_{k}}{M}_{k}}} \right]\mathbb{E}\left[\frac{1}{ {{{P}{t}{\zeta_{k}}{N}_{k}+1}} }\right].
\end{align}
Next, we propose two approximations. The r.v. $M_k$ can be approximated by the r.v. $\widetilde M_k$ with the distribution $Gamma({\widetilde D_{M_k}},{\widetilde \theta _{M_k}})$, where \begin{align}
		\label{d3o3}
		\begin{gathered}
			{\widetilde D_{M_k}} = \frac{{{{\left[ {\left( {{N_t} - K + 1} \right){\varepsilon ^2} + (1 - {\varepsilon ^2})} \right]}^2}}}{{\left( {{N_t} - K + 1} \right){\varepsilon ^4} + {{(1 - {\varepsilon ^2})}^2}}},  \hfill \\
			{\widetilde \theta _{M_k}} = \frac{{\left( {{N_t} - K + 1} \right){\varepsilon ^4} + {{(1 - {\varepsilon ^2})}^2}}}{{\left( {{N_t} - K + 1} \right){\varepsilon ^2} + (1 - {\varepsilon ^2})}}.  \hfill \\ 
		\end{gathered}		
	\end{align}
And the r.v. $N_k$ can be approximated by the r.v. $\widetilde N_k$ with the distribution ${Gamma}(K-1, \frac{(1-\mu_k)(1-\varepsilon^2)}{K-1})$. The proofs of these two approximations are similar to Lemma \ref{Xjinsi} and are omitted for brevity.


Therefore, we can write
  \begin{subequations} 
 	\begin{align}
\mathbb{E}\left[ {{{Y }_{p,k}}} \right]&\nonumber\approx\mathbb{E}\left[ {{Pt\zeta_{k}{\mu_{k}}{\widetilde M}_{k}}} \right]\mathbb{E}\left[\frac{1}{ {{{P}{t}{\zeta_{k}}\widetilde{N}_{k}+1}} }\right]\\&\nonumber=\mathbb{E}\left[ {{Pt\zeta_{k}{\mu_{k}}{\widehat M}_{k}}} \right]\mathbb{E}\left[\frac{1}{ {{{P}{t}{\zeta_{k}}\widetilde{N}_{k}+1}} }\right]\\&\nonumber=\mathbb{E}\left[\frac{{{Pt\zeta_{k}{\mu_{k}}\widehat{M}_{k}}} }{{{{P}{t}{\zeta_{k}}\widetilde{N}_{k}+1}}}\right],
 	\end{align}
 \end{subequations}
where $\widehat M_k$ follows a ${Gamma}(1, \widetilde{D}_{M_k} \widetilde{\theta}_{M_k})$ distribution. Note that the r.v. $\widehat{M}_k$ and the r.v. $\widetilde M_k$ have the same expectation. Introducing $\widehat{M}_k$, which has a shape parameter of 1, is more convenient for our subsequent derivations. Let \(\widetilde{\Gamma}_{p,k}\) be the random variable defined by \(\frac{P t \zeta_k \mu_k \widehat{M}_k}{P t \zeta_k \widetilde{N}_k + 1}\). 

At this point, we note that \( \widetilde{\Gamma}_{p,k} \) has a form completely similar to \( P(1-t)\zeta_{k}\widetilde{Y}_{c,k} \). Therefore, the subsequent proof can follow exactly the same steps as the proof of Lemma \ref{lemma3} to derive the CDF of \( \widetilde{\Gamma}_{p,k} \). Then, based on the obtained CDF and by referring to Lemma \ref{lemma4}, we ultimately derive the expectation of \( \widetilde{\Gamma}_{p,k} \). Since these two steps follow the same reasoning as the previous lemmas, we omit the details for brevity. As a result, we ultimately obtain equation \eqref{mean_Gamma_p_k}.

\ifCLASSOPTIONcaptionsoff
  \newpage
\fi

\bibliographystyle{IEEEtran}

\balance
\bibliography{reference.bib}

\end{document}